\documentclass[11pt,a4paper,onecolumn,accepted=2023-05-07]{quantumarticle}
\pdfoutput=1

\usepackage[usenames,dvipsnames]{xcolor}
\usepackage{amsfonts}
\usepackage{amsmath,amsthm,amssymb,dsfont}
\usepackage{enumitem}
\usepackage[english]{babel}
\usepackage{graphicx}	
\usepackage{float}
\usepackage[caption=false]{subfig}
\usepackage{url}
\usepackage{bbm}
\usepackage{booktabs}
\usepackage[linesnumbered,ruled,vlined]{algorithm2e}

\usepackage{tikz}
\definecolor{myyellow}{RGB}{242,226,149}
\definecolor{grey}{RGB}{150,150,150}
\definecolor{myblue}{RGB}{200,220,230}

\usepackage{pgfplots}
\pgfplotsset{compat=newest}

\usepackage{mathrsfs}

\usepackage{hyperref}
\hypersetup{colorlinks=true,citecolor=blue,linkcolor=blue,filecolor=blue,urlcolor=blue,breaklinks=true}

\usepackage{nicefrac}
\usepackage{mathtools}

\theoremstyle{plain}
\newtheorem{theorem}{Theorem}[section]
\newtheorem{lemma}[theorem]{Lemma}

\newtheorem{corollary}[theorem]{Corollary}
\newtheorem{proposition}[theorem]{Proposition}

\theoremstyle{definition}
\newtheorem{definition}[theorem]{Definition}
\newtheorem{remark}[theorem]{Remark}

\newcommand*{\cA}{\mathcal{A}}
\newcommand*{\cB}{\mathcal{B}}

\newcommand*{\cG}{\mathcal{G}}
\newcommand*{\cH}{\mathcal{H}}
\newcommand*{\cL}{\mathcal{L}}

\newcommand*{\cT}{\mathcal{T}}

\newcommand*{\CC}{\mathbb{C}}

\newcommand*{\eps}{\varepsilon}

\DeclareMathOperator{\tr}{tr}

\newcommand{\1}{\mathbf{1}}

\newcommand*{\id}{\1}

\newcommand*{\<}{\left\langle}
\renewcommand*{\>}{\right\rangle}

\newcommand{\dv}[1]{\frac{\mathrm{d}}{\mathrm{d}#1}}
\newcommand{\ddv}[2]{\frac{\mathrm{d}^{#1}}{\mathrm{d}#2^{#1}}}

\definecolor{mycolor1}{rgb}{0.00000,0.44700,0.74100}%
\definecolor{mycolor2}{rgb}{0.85000,0.32500,0.09800}%
\definecolor{mycolor3}{rgb}{0.92900,0.69400,0.12500}%
\definecolor{mycolor4}{rgb}{0.49400,0.18400,0.55600}%

\renewcommand{\H}{\mathcal{E}}

\renewcommand{\r}{r}

\newcommand{\rhot}{\tilde{\rho}}

\usepackage{mathrsfs}
\newcommand{\Pone}{BHN}

\author{Hamza Fawzi}
\affiliation{Department of Applied Mathematics and Theoretical Physics, University of Cambridge, United Kingdom}
\author{Omar Fawzi}
\affiliation{Univ Lyon, Inria, ENS Lyon, UCBL, LIP, Lyon, France}
\author{Samuel O. Scalet}
\affiliation{Department of Applied Mathematics and Theoretical Physics, University of Cambridge, United Kingdom}

\title{A subpolynomial-time algorithm for the free energy of one-dimensional quantum systems in the thermodynamic limit}

\begin{document}
\maketitle
\begin{abstract}\sloppy
We introduce a classical algorithm to approximate the free energy of local, translation-invariant, one-dimensional quantum systems in the thermodynamic limit of infinite chain size.
While the ground state problem (i.e., the free energy at temperature $T = 0$) for these systems is expected to be computationally hard even for quantum computers, our algorithm runs for any fixed temperature $T > 0$ in subpolynomial time, i.e., in time $O((\frac{1}{\varepsilon})^{c})$ for any constant $c > 0$ where $\varepsilon$ is the additive approximation error. 
Previously, the best known algorithm had a runtime that is polynomial in $\frac{1}{\varepsilon}$.
Our algorithm is also particularly simple as it reduces to the computation of the spectral radius of a linear map.
This linear map has an interpretation as a noncommutative transfer matrix and has been studied previously to prove results on the analyticity of the free energy and the decay of correlations.
We also show that the corresponding eigenvector of this map gives an approximation of the marginal of the Gibbs state and thereby allows for the computation of various thermodynamic properties of the quantum system.
\end{abstract}

\section{Introduction}

Multipartite quantum systems are described by a Hilbert space, which is a tensor product of the single-particle $d$-dimensional spaces. The behaviour of a quantum-many body system is described by a Hamiltonian which models the interaction between the different particles. Of particular interest are $k$-local  Hamiltonians that can be written as a sum of terms acting nontrivially on at most $k$ particles, with $k$ being a constant.  At thermal equilibrium, the system is described by the Gibbs state
\begin{equation}
\label{eq:gibbsintro}
\rho = e^{-\beta H} / Z_{\beta}(H)
\end{equation}
where $\beta = 1/T$ is the inverse temperature, and $Z_{\beta}(H) = \tr\left[e^{-\beta H}\right]$ is the partition function. The free energy of the system at inverse temperature $\beta$ is defined as
\begin{equation}
\label{eq:minFbeta}
F_{\beta}(H) = -\frac{1}{\beta} \log Z_{\beta}(H).
\end{equation}
At zero temperature, i.e., $\beta = +\infty$, $F_{\beta}(H)$ becomes $\lambda_{\min}(H)$, the ground energy of $H$. The problem of computing the ground energy for a given local Hamiltonian is known to be $\mathsf{QMA}$-complete~\cite{kempe2006complexity}, and is the central problem in the area of Hamiltonian complexity~\cite{gharibian2015quantum}. This problem remains $\mathsf{QMA}$-complete even if we restrict ourselves to $2$-local Hamiltonians that are translation-invariant on a chain~\cite{gottesman2009,bausch2017}, i.e., $H = \sum_{i} h_{i,i+1}$ where the operators $h_{i,i+1}$ are given by some Hermitian operator $h$ (the same one for every $i$) acting on particles $i$ and $i+1$.\footnote{Technically, because of the choice of specification of input in~\cite{gottesman2009}, the problem is complete for a scaled version of $\mathsf{QMA}$ called $\mathsf{QMA}_{\mathsf{EXP}}$.}

In order to understand the physical properties of the system at nonzero temperature, it is crucial to understand not only the ground energy, but also the free energy function $F_{\beta}(H)$ as a function of $\beta > 0$ \cite{alhambra2022}. Indeed, computing $F_{\beta}(H)$ and its derivatives with respect to $\beta$ and parameters of the Hamiltonian determines phase transitions and gives access to fundamental physical properties of the system in thermal equilibrium such as the internal energy, specific heat, or magnetic
susceptibility \cite{sandvik2010computational}.

\paragraph{Main result} In this paper, we focus on $2$-local translation-invariant quantum systems on an infinite chain.
As the free energy scales with the system size, in the thermodynamic limit of infinite systems we consider the free energy per particle $f_{\beta}(h)$.
Note that $f_{\beta}(h)$ only depends on the finite matrix $h$ of size $d^2 \times d^2$. Our objective is to design an algorithm to approximate $f_{\beta}(h)$ with a good scaling in terms of the target error $\eps$ and the local dimension $d$. As argued in recent works on Hamiltonian complexity in the thermodynamic limit~\cite{watson2021, aharonov2021}, understanding the dependence of the complexity in terms of the desired precision for infinite systems is often closer to capturing the fundamental problems in many-body physics than understanding the dependence in the system size.
Our main result is an algorithm that given as input $h$ and a target error $\varepsilon$ outputs an approximation of $f_{\beta}(h)$ and of the $k$-particle marginals of the Gibbs state.
\begin{theorem}
\label{thm:mainres}
There is a deterministic algorithm that takes as input a Hermitian operator $h$ acting on $\mathbb{C}^d \otimes \mathbb{C}^d$ satisfying $\| h \| \leq 1$ and $\varepsilon \in (0,1/e)$ and outputs an approximation $\tilde{f}_{\beta}$ satisfying $|\tilde{f}_{\beta} - f_{\beta}(h)| \leq \varepsilon$, where $f_{\beta}(h)$ is the free energy per particle of the infinite translation-invariant Hamiltonian on a chain defined by $h$. For any fixed $\beta > 0$, the running time of the algorithm is $\mathrm{exp} \left( O \left( \log d \frac{\log(1/\eps)}{\log \log(1/\varepsilon)}\right) \right)$.
Moreover, this algorithm can also compute an $\eps$-approximation of the marginal of the Gibbs state on an interval of size $k$ with the same running time for any fixed $k$.
\end{theorem}

Before describing the algorithm and proof method, we make some remarks and discuss related works.

\paragraph{Remarks}
We note that the temperature dependence of the running time is hidden in the $O(.)$ notation as we are interested in the algorithm for fixed temperature. If we want to make the dependence on the inverse temperature $\beta$ explicit, the running time takes the form $\mathrm{exp} \left( O \left( \log d \frac{\log(1/\eps)}{\log \log(1/\varepsilon)}\right)\exp(\exp(O(\beta))) \right)$, where $O(.)$ only hides universal constants.\footnote{The dependence is worse when computing marginals of the Gibbs state, see Corollary~\ref{cor:main}.} 
We note that, from the $\mathsf{QMA}$-hardness results for the ground energy problem~\cite{gottesman2009}, it follows that we should not expect a polynomial dependence in $\beta$ and $1/\eps$ and as such an exponential scaling in $\beta$ is unavoidable in our algorithm.

To appreciate the algorithm we use to prove Theorem~\ref{thm:mainres}, it is instructive to consider first a naive algorithm for this problem sometimes called exact diagonalization. The idea of this algorithm is to consider the Hamiltonian $H_{[1,n]} = \sum_{i=1}^{n-1} h_{i,i+1}$ restricted to only $n$ particles. Computing the free energy per particle $f_{\beta,n}$ of $H_{[1,n]}$ (for any value of $\beta$) can be done in time polynomial in $d^n$ by explicitly writing the $d^n \times d^n$ matrix $H_{[1,n]}$. The sequence $f_{\beta, n}$ does converge to $f_{\beta}$ as $n \to \infty$, but the convergence is in general slow with an error decaying as $\frac{1}{n}$ due to the missing interaction term at the boundary. As a result, for a desired precision $\eps$, we obtain a runtime which is exponential in $1/\eps$. In order to obtain the subpolynomial dependence on $1/\eps$ in Theorem~\ref{thm:mainres}, we need to develop a more refined algorithm.

\paragraph{Related work}
In recent years, there have been multiple works about computing the free energy (or equivalently the partition function) at finite inverse temperature for a given Hamiltonian. In particular, for local Hamiltonians on an arbitrary bounded-degree graph, algorithms have been developed in \cite{kuwahara2020,harrow2020,mann2021} with performance guarantees when the inverse temperature $\beta$ is below some critical inverse temperature. The runtime of these algorithms is polynomial or quasi-polynomial in the number of particles and in $1/\eps$. 
These works rely on the so called cluster expansion, which, at its core, is a Taylor expansion of the partition function at $\beta = 0$.
Truncating this expansion at a certain order and bounding the remainder terms 
allows for the approximation of the free energy for $\beta$ small enough.
Indeed, the sum of remainder terms no longer converges if $\beta$ is too large which introduces a critical inverse temperature above which such algorithms do not have convergence guarantees.

However, a different method was used in~\cite{kuwahara2018} to obtain an algorithm for all temperatures for one-dimensional finite quantum systems.
This algorithm combines several results from the analysis of 1D-systems: quantum belief propagation~\cite{hastings2007}, together with a locality result about Gibbs states which was only recently proven in general~\cite{bluhm2022exponential}. For a system of $n$ particles in 1D, the running time of the algorithm is $n (\frac{1}{\eps})^{O(1)}$, where $O(1)$ is a constant that depends exponentially on $\beta$. This algorithm can readily be applied to the infinite translation-invariant chain by setting $n = 1/\eps$ and this leads to an algorithm that is polynomial in $1/\eps$. For its implementation, the algorithm involves several choices of length scales to ensure convergence and numerical integrations to obtain the operators from the belief propagation.

Another algorithm, based on tensor networks, for computing thermal expectation values for a system of $n$ particles was designed in \cite{kuwahara2021,alhambra2021} and has a runtime with an improved dependence on $\eps$, namely $n e^{\tilde{O}(\sqrt{\log(n/\eps)})}$, where the $\tilde{O}$ depends linearly on $\beta$. As with the algorithm of~\cite{kuwahara2018}, it can be readily applied to the infinite translation-invariant chain by setting $n=1/\eps$. This leads to an algorithm that is quasi-linear in $1/\eps$ for computing thermal expectation values for 
infinite translation-invariant chains. This algorithm can be used to compute the free energy using the argument in \cite[Lemma 12]{bravyi2021} at the cost of an additional polynomial overhead in $1/\eps$, thus leading to an algorithm for computing the free energy for infinite translation-invariant chains that is polynomial in $1/\eps$. 

In a different regime, classical algorithms have been designed in~\cite{bravyi2021} to compute the free energy of dense Hamiltonians based on convex relaxations. These algorithms have a runtime that is exponential in $1/\eps$.

Besides this line of work on provably convergent algorithms to which our work shall also contribute, there are numerous algorithms that effectively address the problem despite having no convergence results or only in special cases.
For the free energy problem this includes most notably Quantum Monte Carlo methods \cite{troyer2003,suzuki1993}.
These probabilistic algorithms lack rigorous results on their runtime except for few special cases and are known to fail for Hamiltonians that have the so-called sign problem.
Another example of effective algorithms for the ground state energy problem ($\beta = +\infty$) in one dimension are tensor networks and the DMRG algorithm \cite{white1992, schollwock2011}.
The convergence of a related algorithm to the ground state energy has been proven under the additional assumption that the Hamiltonian is gapped \cite{landau2013}.

\paragraph{Proof technique}
Before giving an overview of the algorithm establishing Theorem~\ref{thm:mainres}, it is worth mentioning that the analogous classical problem has a very simple solution. In fact, using the technique of transfer matrices (see e.g. \cite{friedli2017}), for any $\beta$ the free energy per particle can be obtained from the eigenvalue of a simple $d\times d$ matrix and thus the problem reduces to standard numerical algorithms applied to some fixed matrix. This implies very efficient algorithms for any $\beta$ including $\beta=+\infty$.

However, the quantum case is significantly more complicated. This is illustrated for example by the fact that, when $\beta = +\infty$, the problem is $\mathsf{QMA}$-hard, and also that the simple Markov property for classical Gibbs states in one-dimension does not hold in the quantum setting.
In his seminal work, Araki~\cite{araki1969} proposed a quantum analogue of the transfer matrix, but it is a linear map between infinite-dimensional spaces. In our algorithm, we use a finite-dimensional approximation to this map.
Our main technical result is to prove that the spectral radii of these finite-dimensional approximations converge superexponentially fast  to $e^{-\beta f_{\beta}(h)}$. The algorithm (see Algorithm~\ref{algo:main}) is then simply to choose the finite-dimensional approximation parameter $L$ for the transfer matrix as a function of the desired precision $\eps$ and then compute the spectral radius of the corresponding linear map.

\begin{algorithm}[!ht]
\SetKwInput{KwInput}{Input}
\SetKwInput{KwRequire}{Parameters}
\SetKwInput{KwOutput}{Output}
\caption{Algorithm for the computation of the free energy}
\KwRequire{Inverse temperature $\beta$, universal constant $C$}
\KwInput{$d$ local dimension, Hamiltonian term $h\in\mathbb C^{d^2 \times d^2}$ such that $\|h\|\le1$, error $\eps$}
\KwOutput{$\tilde{f}_{\beta}$ approximation to the free energy $f_{\beta}(h)$}
$L \gets \log(1/\eps)\exp(\exp(C(\beta+1)))/\log(\log(1/\eps))$\tcc*[r]{parameter for approximation}
\tcc{matrix representation of a linear map from $\mathbb{C}^{d^{L-1} \times d^{L-1}}$ to itself:}
$\cL_L^*(\cdot) \gets \tr_{L} \left( e^{-\beta H_{[1,L]}/2} e^{\beta H_{[2,L]}/2} (\1 \otimes \cdot) e^{\beta H_{[2,L]}/2} e^{-\beta H_{[1,L]}/2} \right)$\;
$r_{L} \gets$ spectral radius of $\cL_L^*$ \;
$\tilde{f}_{\beta} \gets - (\log r_{L})/\beta$ \;
\caption{Algorithm for computing the free energy per particle. The constant $C$ is a number that can be obtained from our proofs. See Section~\ref{sec:algo} for more precise definitions.}
\label{algo:main}
\end{algorithm}

To analyse the algorithm we make extensive use of Araki's expansionals~\cite{araki1969} to show that the marginal $\rho_L$ on the first $L-1$ sites of the infinite Gibbs state, is an \emph{approximate eigenvector} of the finite-dimensional map $\cL_L^*$, i.e., that $\|\cL_L^*(\rho_L) - e^{-\beta f_{\beta}(h)} \rho_L\|_1$ decays superexponentially fast in $L$. By using variational expressions of the spectral radius of positive maps (so called Collatz-Wielandt formula), this allows us to show that the spectral radius of $\cL_L^*$ is superexponentially close to $e^{-\beta f_{\beta}(h)}$.
We note that standard perturbation bounds for eigenvalues of non-normal operators have a very bad dependence on dimension, and are thus not usable here, see e.g., \cite[Chapter VIII]{bhatia}.
To prove that the corresponding eigenvector of $\cL_L^*$ is close to $\rho_L$, we establish a quantitative \emph{primitivity} condition for $\cL_L^*$, i.e., we prove that a sufficiently high power of $\cL_L^*$ maps nonzero positive semidefinite operators to positive \emph{definite} ones. Using tools from the Perron-Frobenius theory of positive operators ---more precisely the Hilbert projective metric---, this allows us to show that $\rho_L$ is superexponentially close to the eigenvector of $\cL_L^*$ associated to its spectral radius.

Let us also mention that the above techniques based on~\cite{araki1969} are specific to one dimension.
In particular, the results in there are related to the absence of thermal phase transitions, which do occur in higher dimensions so an extension of this approach to higher dimension is not possible. See also \cite{sly2010,sly2012} for hardness results of the classical partition function on bounded degree graphs.

\paragraph{Numerical implementation} We also implement our algorithm and run it on a Hamiltonian for which the free energy function is known exactly.
We observe very small errors (machine precision) already for moderate choices of $L$.
We also observe that for this example the scaling of the error with inverse temperature is better than the worst-case estimates derived theoretically.

\paragraph{Organization} The paper is structured as follows, we first give in Section~\ref{sec:prelim} the mathematical definitions in which the problem is formulated and review some technical results needed in our proof.
We then present in Section~\ref{sec:algo} the statements of our main results including error bounds for the computed free energy and Gibbs state as well as a bound on the runtime of the corresponding algorithm. Sections~\ref{sec:proof-eigenvalue} to \ref{sec:proof-runtime} contain the proofs of our main results.
In Section~\ref{sec:numerics} we present our numerical results.

\section{Preliminaries}
\label{sec:prelim}

We start by introducing some notation for the description of quantum many-body systems and recap results from \cite{araki1969} needed for the proofs in this manuscript.

\subsection{Setup and basic notations}

We consider a quantum 1D chain on the half-line $[1,\infty)$. Let $d\ge2$ be the local dimension, $\cH_i\cong\mathbb C^d$ the local Hilbert space, and $\cH_{[a,b]}=\bigotimes_{i\in[a,b]}\cH_i$ be the Hilbert space associated to sites $[a,b]$, where $b \geq a$ are integers.
We let $\cA_{[a,b]} = \mathcal B(\cH_{[a,b]})$ be the associated algebra of operators acting on $\cH_{[a,b]}$. If $I \subset J$ are two intervals in $\mathbb{Z}$, we will often identify elements of $\cA_I$ as elements of $\cA_J$ by tensoring with the identity on $J\setminus I$.
The adjoint map is the partial trace over the sites in $J\setminus I$, $\tr_{J\setminus I}:\cA_J\to\cA_I$.

Let $h \in \CC^{d^2 \times d^2}$ be the two-body Hamiltonian acting on two sites\footnote{Note that without loss of generality we restrict to 2-local Hamiltonians for simplicity.
A more general $r$-local Hamiltonian can be reduced by blocking $r/2$ sites into a single local Hilbert space of dimension $d^{(r/2)}$ and collecting the Hamiltonian terms.
This recast Hamiltonian can then be used as an input to the algorithm.}, and let $h_{i,i+1} \in \cA_{[i,i+1]}$ the corresponding Hamiltonian when acting on sites $\{i,i+1\}$. If $a \leq b$, we let 
\[
H_{[a,b]} = \sum_{i=a}^{b-1} h_{i,i+1} \in \cA_{[a,b]}.
\]
For an inverse temperature $\beta > 0$ and  any integer $N > 1$, the \emph{partition function} is defined as
\[
Z_{\beta,N}(h) = \tr \exp\left(-\beta H_{[1,N]}\right).
\]
The infinite translation-invariant free energy per particle is defined as:
\begin{equation}\label{eq:f}
f_{\beta}(h) = \lim_{N\to \infty} -\frac{1}{\beta N} \log Z_{\beta,N}(h).
\end{equation}
This limit is known to exist, see e.g., \cite[Theorem 15.5]{ohya2004}. 
By appropriately rescaling $h$, one can assume without loss of generality that $\beta = 1$; indeed, it is immediate to check that $f_{\beta}(h) = (1/\beta) f_{1}(\beta h)$. In the rest of the paper we will thus assume that $\beta = 1$, and omit the subscript $\beta$ from the definition $f_{\beta}(h)$. We will also suppress the dependence on $h$ in the notation, as it will be clear from the context, and just write
\[
Z_N=Z_{1,N}(h) \quad \text{ and } \quad f=f_1(h).
\]

To keep track of the complexity of our algorithm in the inverse temperature, we make the following convenient definitions:

\begin{definition}
\label{def:Ponecst}
We say that a constant $\cG \geq 0$ is \emph{bounded in Hamiltonian norm (\Pone)} if there exist some numerical constants $C_1,C_2>0$ such that
\[
\cG\le \exp(\exp(C_1\exp(C_2\|h\|))).
\]
We say that a function $\varepsilon(L) \geq 0$ is \emph{superexponentially decaying} if there exist some numerical constants $C_1,C_2,C_3$ such that
\[
\varepsilon(L)\le \exp(\exp(C_1\exp(C_2\|h\|))) \frac{\exp(C_3 \|h\|L)}{(L/2)!}.
\]
\end{definition}

It is immediate to check that products of $\Pone$-constants have property $\Pone$ and that the product of a $\Pone$-constant and a superexponentially decaying function is superexponentially decaying.

\subsection{The map \texorpdfstring{$\cL_L^*$}{}}

For an integer $L \geq 2$, define
\begin{equation}
\label{eq:defELm}
\begin{aligned}
E_L &= e^{-H_{[1,L]}/2} e^{H_{[2,L]}/2} \in \cA_{[1,L]}
\end{aligned}
\end{equation}
and let $\cL_L^*:\cA_{[1,L-1]} \to \cA_{[1,L-1]}$ be the positive linear map defined by
\begin{equation}
\label{eq:defLLstar}
\cL_L^*(Q) = \tr_L(E_L (\1 \otimes Q) E_L^{\dagger}).
\end{equation}
In the classical or commuting case, the operators $E_L$ are independent of $L$ and act nontrivially only on particles $1$ and $2$, and the linear map $\cL_L^*$ plays the role of a \emph{transfer matrix}. For the general quantum case, it is a highly nontrivial result of \cite{araki1969} that the operator $E_L$ is close to an operator with support only on the first sites, that it is bounded uniformly in $L$, and that it converges for $L\to\infty$. 
The map $\cL_L^*$ can then be thought of as a finite-dimensional approximation of an infinite-dimensional ``noncommutative transfer matrix'' introduced by Araki \cite{araki1969}.\footnote{By extending the maps $\cL_L^*$ to an infinite-dimensional setting a limit of these maps can be defined. In \cite{araki1969,perezgarcia2020} it has been shown that its eigenvector and eigenvalue are a thermal state and the free energy of the system respectively, regaining the meaning of transfer matrices as in the classical case.}
The key quantity of interest to us will be the \emph{spectral radius} of $\cL^*_L$, which we will denote by $r_L$
\begin{equation}
r_L = \max\{|\lambda| : \lambda \text{ is an eigenvalue of $\cL_L^*$}\}.
\end{equation}

The following lemma collects the results from \cite{araki1969} that we will need for the analysis of our algorithm. Note that, for convenience, we use the formulation of \cite{perezgarcia2020} which mostly considers exponentially decaying Hamiltonians instead of strictly local ones but still recovers the following results when restricting to local interactions. 
The dependency of the constants on $\|h\|$ might not be immediately clear from the formulation in \cite{perezgarcia2020}, so we discuss how to obtain the claimed dependencies in Appendix \ref{sec:constants}.
\begin{lemma}[{{\cite[ Proposition 4.2]{perezgarcia2020}}}]\label{lem:arakiE}
There exists a $\Pone$-constant $\cG$ and a superexponentially decaying function $\varepsilon(L)$ such that for any $L \geq 2$
\begin{enumerate}[label=(\roman*)]
\item $\|E_L\|,\|E_L^{-1}\|\le\cG$
\item for $L\le M$ we have $\|E_L-E_M\|,\|E_L^{-1}-E_M^{-1}\|\le\varepsilon(L)$
\end{enumerate}
From the above it also follows immediately (by renaming constants) that
\begin{enumerate}[label=(\roman*)]
\setcounter{enumi}{2}
\item for $L\le M$ we have $\|E_LE_M^{-1}-\1\|\le\varepsilon(L)$
\end{enumerate} 
\end{lemma}

\subsection{Gibbs states} 

We also introduce marginals of thermal states on finite systems and in the thermodynamic limit.
For $m > L \geq 2$, these are defined as
\begin{equation}
\label{eq:defrhoLm}
\begin{aligned}
\bar\rho_{L,m} &= \tr_{[L,m]} e^{-H_{[1,m]}} \in \cA_{[1,L-1]}\\
\rho_{L,m} &= \bar\rho_{L,m} / \tr \bar\rho_{L,m}.
\end{aligned}
\end{equation}
In \cite{araki1969} (see also \cite[Lemma 4.15]{perezgarcia2020}) it was shown that the limit $\lim_{m\to\infty} \rho_{L,m}$ exists, which we denote by $\rho_L$:
\begin{equation}
\label{eq:rhoL}
\rho_L = \lim_{m\to \infty} \rho_{L,m}.
\end{equation}

\section{Description of the algorithm and overview of the analysis}
\label{sec:algo}

In this section, we present our main convergence results and the resulting runtime of the proposed algorithm.

Given a 2-body Hamiltonian $h \in \CC^{d^2\times d^2}$, our first main result shows that the sequence $-\log r_L$ (where  $r_L$ is the spectral radius of the map $\cL_L^*$ introduced in \eqref{eq:defLLstar}) converges superexponentially fast to the free energy per site $f=f(h)$ of the infinite chain, namely:
\begin{theorem}
\label{thm:main}
There is a superexponentially decaying function $\varepsilon(L)$ such that $|(-\log r_L) - f| \leq \varepsilon(L)$ for all $L \geq 2$.
\end{theorem}

In order to compute other properties of the Gibbs state such as expectation values of local observables or correlation functions, it is useful to also have a description of the Gibbs state.
It turns out that the eigenvector corresponding to the spectral radius of $\cL_L^*$ approximates the marginal $\rho_{L}$ of the Gibbs state defined in \eqref{eq:rhoL}. This is the object of the next theorem (recall the definition of a $\Pone$-constant in Def. \ref{def:Ponecst}):
\begin{theorem}
\label{thm:mainState}
There exist $\Pone$-constants  $\cG,\cG'$ and a superexponentially decaying function $\varepsilon(L)$ such that the following is true. 
Let $v_L\in\cA_{[1,L-1]}$ be the eigenvector of the map $\cL^*_L$ corresponding to its spectral radius. For $L\ge e^{\cG}$, $v_L$ is unique, strictly positive, and
\[
\|v_L-\rho_L\|_1 \leq e^{\cG'} \varepsilon(L).
\]
\end{theorem}

Note that the above theorem proves convergence to the one-sided Gibbs state, i.e., the Gibbs state of a chain that extends to infinity only in one direction. 
We focus on this case to simplify the presentation. The case of the chain that is infinite in both directions can be handled in the exact same way by using the two-sided version of the noncommutative transfer operator as defined in \cite{araki1969}.
Alternatively, the one-sided version of the algorithm can be used in a black box way to obtain the two-sided marginal by recasting the Hamiltonian as we explain in Appendix~\ref{sec:recastH}.
We also note that one can compute the marginals of the Gibbs state by computing certain directional derivatives of the free energy function. This is described in more detail in Appendix~\ref{sec:fToP} and it directly yields the marginals for the two-sided chain, as the free energy is the same in both cases. It is, however, inefficient to use this method to compute marginals of many particles.

Theorems \ref{thm:main} and \ref{thm:mainState} imply that the free energy and the $k$-particle marginals of the Gibbs state can be approximated up to any given error by choosing an appropriate $L$, and computing the spectral radius of $\cL_L^*$ and the corresponding eigenvector. 
By making this choice and its dependence on the problem input explicit, we obtain the following corollary.

\begin{corollary}\label{cor:main}
For some numerical constant $C>0$, there exists an algorithm that takes as input the local dimension of a quantum system $d$, its 2-local translation-invariant Hamiltonian given by $h$, an additive error $\varepsilon<1/e$, runs in time at most
\begin{equation}
\label{eq:complexitymain}
\exp\left(\log(d)\exp(\exp\left(C\left(\|h\|+1\right)\right))\frac{\log(1/\varepsilon)}{\log(\log(1/\varepsilon))}\right)\ ,
\end{equation}
and outputs an approximation $\tilde f$ such that $|f(h)-\tilde f|\le\varepsilon$, where $f(h)$ is defined in \eqref{eq:f}.

Furthermore, for some numerical constants $C_1,\ldots,C_5>0$ and a given system size $k$, there exists an algorithm that takes the same inputs as above, runs in time at most
\begin{equation}
\label{eq:complexityevector}
\exp\left(\log(d)\max\left\{\exp\left(\exp\left(C_1e^{C_2\|h\|}\right)\right)\frac{\log(1/\varepsilon)}{\log(\log(1/\varepsilon))},\exp\left(e^{C_3e^{e^{C_4\|h\|}}}\right), C_5k\right\}\right)\ ,
\end{equation}
and outputs an approximation $v\in\cA_{[1,k]}$ such that $\|v-\rho_{k+1}\|_1\le\varepsilon$, where $\rho_{k+1}$ is defined in \eqref{eq:rhoL}.
\end{corollary}

While we have made both the dependence on $\varepsilon$ and the dependence on $\|h\|$ explicit,  we want to point out that our algorithm should be considered for some fixed bound on $\|h\|$ (recall that we have assumed $\beta = 1$, for general $\beta$, this value is the inverse temperature times the norm of the Hamiltonian).
As pointed out in the introduction, the ground state problem for the system we consider is computationally hard.
If an algorithm was able to solve the free energy problem efficiently in $1/\varepsilon$ and $\|h\|$, it could also approximate the ground state energy efficiently by taking $\|h\|~\to~\infty$. We formulate and prove such a reduction in Appendix~\ref{sec:GSReduction}, i.e., we establish the $\mathsf{QMA_{EXP}}$-hardness of the infinite translation-invariant free energy problem with the temperature as an additional problem input. This shows that, unless $\mathsf{QMA_{EXP}} = \mathsf{EXP}$, no algorithm can have a running time of the form $\exp(\mathrm{polylog}(\| h \|, 1/\eps))$.

We note that the temperature dependence in the second part of the corollary is worse than in the first part\footnote{To give more insight into the complexity estimate \eqref{eq:complexityevector}, the second term in the maximum in~\eqref{eq:complexityevector} comes from the lower bound on $L$ in Theorem~\ref{thm:mainState}.
The last term comes from the fact that we have to compute the eigenvector of at least the map $\cL^{*}_{k+1}$ for it to be defined on at least $k$ sites (for larger maps we take a partial trace in the end).}, which could be an artefact of our proof method (see Remark~\ref{rem:conjectureprimitivity}).
In fact, it is also possible to compute thermodynamic quantities as derivatives of the free energy with respect to parameters in the Hamiltonian.
This allows us to just use the first part of our proposed algorithm to access other properties of the Gibbs state. The details of this approach are worked out in detail in  Appendix~\ref{sec:fToP} for 2-local observables. 
In this case, this slightly improves the temperature dependence,
as it removes the second term in \eqref{eq:complexityevector}, while still having the same dependence as the first term.
The reason for this is that the error bounds for approximations of the derivative involve the second derivative of the free energy, which can only be bounded by invoking results from \cite{araki1969} as it is related to the analyticity of the free energy.
We note however that this method is not practical to compute the average of $k$-local observables for moderately large $k > 2$, as it requires to block the system first which results in a sytem with local dimension $d^{k/2}$ and interaction strength up to $(k/2) \|h\|$. On the other hand, our algorithm allows us to get $L$-local marginals for free. In addition, as illustrated in the numerical results Section~\ref{sec:numerics}, our algorithm also allows us to compute entropic quantities on the marginals such as the (conditional) mutual information.

\section{Convergence speed of the spectral radius of \texorpdfstring{$\mathcal{L}_{L}^*$}{} (proof of Theorem~\ref{thm:main})}
\label{sec:proof-eigenvalue}

In this section, we prove Theorem~\ref{thm:main}, showing that $-\log r_L$ converges superexponentially fast in $L$ to the free energy $f$, where $r_L$ is the spectral radius of the finite-dimensional linear map $\cL_L^*$ defined in \eqref{eq:defLLstar}. The key to our proof is to show that $\rho_L$, the marginal of the Gibbs state on $[1,L-1]$ (see \eqref{eq:defrhoLm}), is an approximate eigenvector of $\cL_L^*$, in the sense that
\begin{equation}
\label{eq:rhoLapxevec}
(1-\eps(L)) e^{-f} \rho_L \leq \cL_L^*(\rho_L) \leq (1+\eps(L)) e^{-f} \rho_L
\end{equation}
where $\eps(L)$ is a superexponentially decaying function in $L$. This will be proven in Corollary~\ref{cor:LFPSandwich}.

We start by giving an alternative expression for the free energy. While the limit of the normalized log-partition function exists, this does not immediately imply that the unnormalized partition function grows in approximately constant steps.
It can, however, be shown in the case of a translation-invariant interaction, using locality results of Gibbs states:
\begin{lemma}\label{lem:fRatioLimit}
The free energy defined in \eqref{eq:f} also admits the following expression:
\begin{equation}
\label{eq:f2}
f = -\lim_{N\to \infty} \log \frac{Z_{N+1}}{Z_N}.
\end{equation}
Moreover $e^{-f},e^f\le e^{\|h\|}$.
\end{lemma}
\begin{proof}
Using the continuity of the logarithm, we have to prove the convergence of $Z_{N+1}/Z_N$.
For $N > M$, we have
\begin{align*}
\frac{Z_{N+1}}{Z_N}&=\frac{\tr[e^{-H_{[1,N+1]}}]}{\tr[e^{-H_{[1,N]}}]}=\frac{\tr[\id\otimes e^{-H_{[1,N]}}E_{N+1}^\dagger E_{N+1}]}{\tr[e^{-H_{[1,N]}}]}\\
&=\frac{\tr[e^{-H_{[2,N+1]}}E_{M+1}^\dagger E_{M+1}]}{\tr[e^{-H_{[2,N+1]}}]}-\frac{\tr[\id\otimes e^{-H_{[1,N]}}(E_{M+1}^\dagger E_{M+1}-E_{N+1}^\dagger E_{N+1})]}{\tr[e^{-H_{[1,N]}}]}\\
&=\tr[\rho_{M+1}\tr_1[E_{M+1}^\dagger E_{M+1}]]+\tr[(\rho_{M+1,N}-\rho_{M+1})E^\dagger_{M+1}E_{M+1}]\\
&\quad-\tr[\id\otimes \rho_{N+1,N}(E_{M+1}^\dagger E_{M+1}-E_{N+1}^\dagger E_{N+1})].
\end{align*}
Thus for $N,N' > M$ we have
\[
\begin{aligned}
\frac{Z_{N+1}}{Z_N}-\frac{Z_{N'+1}}{Z_{N'}} &= \tr[(\rho_{M+1,N} - \rho_{M+1,N'}) E_{M+1}^{\dagger} E_{M+1}]\\
& \quad - \tr[\id\otimes \rho_{N+1,N}(E_{M+1}^\dagger E_{M+1}-E_{N+1}^\dagger E_{N+1})]\\
& \quad + \tr[\id\otimes \rho_{N'+1,N'}(E_{M+1}^\dagger E_{M+1}-E_{N'+1}^\dagger E_{N'+1})].
\end{aligned}
\]
Let $\eps > 0$. By Lemma \ref{lem:arakiE}, there is a large enough $M$ such that for any $N,N' > M$ the last two terms have magnitude $\leq \eps$.
Furthermore, since the sequence $\rho_{M+1,N}$ converges as $N \to \infty$ (see \eqref{eq:rhoL}), there is a large enough $N_0$ such that for any $N,N' \geq N_0$ $\|\rho_{M+1,N} - \rho_{M+1,N'}\|_1 \leq \eps$.
Together with the boundedness of $E^\dagger_{M+1}E_{M+1}$ from Lemma~\ref{lem:arakiE}, this tells us that the sequence $(Z_{N+1}/Z_N)$ is Cauchy and thus converges.

The moreover part follows from
\[
\left|\log\left(\frac{Z_{m+1}}{Z_m}\right)\right|\le\|h\|
\]
by \cite[Lemma 3.6]{lenci2005}.
\end{proof}

The second technical lemma we need gives a bound on the operator norm of $\rho_{L,m}$, the marginal of the finite Gibbs state defined in \eqref{eq:defrhoLm}, and its inverse.

\begin{lemma}
\label{lem:cond}
There is a $\Pone$-constant $\cG$ such that for any $L,m$, $\|\rho_{L,m}\|,\|\rho_{L,m}^{-1}\| \leq \cG e^{2\|h\|L}$.
\end{lemma}
\begin{proof}
We consider the operator 
\[
\tilde{E}_{L,m} = e^{-H_{[1,m]}/2} e^{H_{[L,m]}/2} e^{H_{[1,L-1]}/2}.
\]
By \cite[Proposition 4.2]{perezgarcia2020}, we know that $\cG^{-1} \1 \leq \tilde{E}_{L,m}^{\dagger} \tilde{E}_{L,m} \leq \cG \1$ where $\cG > 0$ is some $\Pone$-constant. 
This implies that
\[
\cG^{-1} e^{-H_{[1,L-1]}} \leq e^{H_{[L,m]}/2} e^{-H_{[1,m]}} e^{H_{[L,m]}/2} \leq \cG e^{-H_{[1,L-1]}}
\]
and, since $e^{-\|h\|L} \1 \leq e^{-H_{[1,L-1]}} \leq e^{\|h\|L} \1$, we then get
\[
\cG^{-1} e^{-\|h\|L} e^{-H_{[L,m]}} \leq e^{-H_{[1,m]}} \leq \cG e^{\|h\|L} e^{-H_{[L,m]}}.
\]
Taking the partial trace (which is a positive map), we get, with $\gamma = \tr e^{-H_{[L,m]}} > 0$, and recalling that $\bar\rho_{L,m} = \tr_{[L,m]} e^{-H_{[1,m]}}$,
\[
\cG^{-1} e^{-\|h\|L} \gamma \1 \leq \bar\rho_{L,m} \leq \cG e^{\|h\|L} \gamma \1.
\]
We then divide by $Z_m$ which normalizes $\bar\rho_{L,m}$
\[
\cG^{-1} e^{-\|h\|L} \frac\gamma{Z_m}\1 \leq \frac{\bar\rho_{L,m}}{Z_m} \leq \cG e^{\|h\|L} \frac\gamma{Z_m}\1.
\]
Using translation invariance, the factor $\gamma/Z_m$ can be bounded as follows
\begin{equation}\label{eq:partFuncBound}
\left|\log\left(\frac{\gamma}{Z_m}\right)\right|=\left|\log(Z_{m-L+1})-\log(Z_m)\right|\le\|H_{[1,L-1]}\|\le \|h\|(L-1)
\end{equation}
by \cite[Lemma 3.6]{lenci2005}, so since $\rho_{L,m} = \bar \rho_{L,m} / Z_m$
\[
\cG^{-1} e^{-2\|h\|L} \1 \leq \rho_{L,m} \leq \cG e^{2\|h\|L} \1.
\]
or equivalently $\|\rho_{L,m}\|, \|\rho_{L,m}^{-1}\|\le\cG e^{2\|h\|L}$ as desired.
\end{proof}

The next lemma is crucial, and shows that $\rho_L$ is an approximate eigenvector of the linear map $\cL_L^*$ with eigenvalue $e^{-f}$, for the trace distance. 

\begin{lemma}
There is a superexponentially decaying function $\varepsilon(L)$ such that
\begin{equation}
\label{eq:b1}
\left\|\cL^*_L(\rho_L) - e^{-f} \rho_L\right\|_1 \leq \varepsilon(L).
\end{equation}
\end{lemma}
\begin{proof}
For any $m > L$, we have, with $X = E_L E_{m+1}^{-1} \in \cA_{[1,m+1]}$
\[
\begin{aligned}
\cL^*_L(\bar\rho_{L,m}) &= \tr_L(E_L \tr_{[L+1,m+1]} e^{-H_{[2,m+1]}} E_L^{\dagger})\\
&= \tr_{[L,m+1]}(E_L e^{-H_{[2,m+1]}} E_L^{\dagger})\\
&= \tr_{[L,m+1]}(XE_{m+1} e^{-H_{[2,m+1]}} E_{m+1}^{\dagger} X^{\dagger})\\
&= \tr_{[L,m+1]}(X e^{-H_{[1,m+1]}} X^{\dagger}).
\end{aligned}
\]
By Lemma~\ref{lem:arakiE} we have $\|X - \1\| \leq \varepsilon(L)$ for some superexponentially decaying function $\varepsilon(L)$ and $\|X\|\le\cG^2$ for some $\Pone$-constant $\cG$.
Let $Z_m = \tr e^{-H_{[1,m]}}$ so that $\rho_{L,m} = \bar\rho_{L,m} / Z_m$. Let $S$ be an arbitrary Hermitian operator in $\cA_{[1,L-1]}$. Then, interpreting as necessary $S \in \cA_{[1,m+1]}$ by tensoring with the identity and denoting the Hilbert-Schmidt inner product as $\<A,B\>=\tr[A^*B]$, we have
\begin{equation}
\label{eq:b4}
\begin{aligned}
\< \cL^*_L(\rho_{L,m}) , S \> &= \frac{1}{Z_m} \< e^{-H_{[1,m+1]}} , X^{\dagger} S X \>\\
&= \frac{Z_{m+1}}{Z_m} \< \frac{e^{-H_{[1,m+1]}}}{Z_{m+1}} , X^{\dagger} S X \>\\
&= \frac{Z_{m+1}}{Z_m} \< \rho_{L,m+1} , S \> + \frac{Z_{m+1}}{Z_m} \< \frac{e^{-H_{[1,m+1]}}}{Z_{m+1}} , X^{\dagger} S X - S \>.
\end{aligned}
\end{equation}
Note that 
\[
\begin{aligned}
\|X^{\dagger} S X - S\| &= \|(X^{\dagger} - \1) S X + S(X-\1)\|\\
&\leq \|X-\1\|\|S\|(\|X\|+1) \leq \varepsilon'(L) \|S\|
\end{aligned}
\]
where $\varepsilon'(L) = \varepsilon(L)(1+\cG^2)$ is a superexponentially decaying function. Since $Z_{m+1}/Z_m$ is bounded by $\exp(\|h\|)$ (see Equation \eqref{eq:partFuncBound}) and $\tr \rho_{L,m+1} = 1$, the second term in the last line of \eqref{eq:b4} has a magnitude which is at most $\varepsilon''(L) \|S\|$, where $\varepsilon''(L)$ is a superexponentially decaying function. Thus this shows that for any $S$, we have
\[
\left| \< \cL_L^*(\rho_{L,m}) - \frac{Z_{m+1}}{Z_m} \rho_{L,m+1} , S \> \right| \leq \varepsilon''(L) \|S\|
\]
which is the same as
\[
\|\cL_L^*(\rho_{L,m}) - \frac{Z_{m+1}}{Z_m} \rho_{L,m+1} \|_1 \leq \varepsilon''(L).
\]
Letting $m\to \infty$ gives us \eqref{eq:b1} as desired.
\end{proof}

The following corollary shows that $\rho_L$ is an approximate eigenvector  of $\cL_L^*$ in the positive semidefinite order.

\begin{corollary}\label{cor:LFPSandwich}
There is a superexponentially decaying function $\varepsilon(L)$ such that
\[
(1-\varepsilon(L)) e^{-f} \rho_L \leq \cL_L^*(\rho_L) \leq (1+\varepsilon(L)) e^{-f} \rho_L.
\]
\end{corollary}
\begin{proof}
From \eqref{eq:b1} we get, 
\[
\begin{aligned}
\|\cL_L^*(\rho_L) - e^{-f} \rho_L\| \leq \|\cL_L^*(\rho_L) - e^{-f} \rho_L\|_1
\leq \varepsilon_1(L)
\end{aligned}
\]
The above can be rewritten as
\[
- \varepsilon_1(L) \1 \leq \cL_L^*(\rho_L) - e^{-f} \rho_L \leq \varepsilon_1(L)  \1.
\]
Writing $\1 \leq \|\rho_L^{-1}\| \rho_L$ the above implies
\[
- \varepsilon_1(L) \|\rho_L^{-1}\| \rho_L \leq \cL_L^*(\rho_L) - e^{-f} \rho_L \leq \varepsilon_1(L) \|\rho_L^{-1}\| \rho_L.
\]
Using Lemma~\ref{lem:cond} to bound $\|\rho_L^{-1}\|$ and $e^f\le e^{\|h\|}$, we obtain that $\varepsilon(L) := e^f \varepsilon_1(L) \|\rho_L^{-1}\|$ is still superexponentially decaying, and we have
\begin{equation}
\label{eq:LLQineq}
e^{-f}(1-\varepsilon(L)) \rho_L \leq \cL_L^*(\rho_L) \leq e^{-f}(1+\varepsilon(L)) \rho_L.
\end{equation}
\end{proof}

We can now finish the proof of Theorem~\ref{thm:main}.

\begin{proof}[Proof of Theorem~\ref{thm:main}]
Since $\cL_L$ is a positive map with the same spectral radius as $\cL_L^*$, there exists $w_L \geq 0$ such that $\cL_L(w_L) = r_L w_L$ \cite[Theorem 2.5]{evans1978}. 
Then, we have
\begin{equation}
\label{eq:collatzwielandt}
r_L \<w_L, \rho_L\> = \<\cL_L(w_L), \rho_L\> = \<w_L, \cL_L^*(\rho_L)\> \leq e^{-f}(1+\varepsilon(L)) \<w_L, \rho_L\>.
\end{equation}
Since $\<w_L,\rho_L\> > 0$ (because $w_L \geq 0$ and $\rho_L > 0$), we get from the above that $r_L \leq e^{-f}(1+\varepsilon(L))$. 
In a similar way, we get $r_L \geq e^{-f}(1-\varepsilon(L))$, and so
\[
e^{-f}(1-\varepsilon(L)) \leq r_L \leq e^{-f}(1+\varepsilon(L)),
\]
as desired.
\end{proof}

\begin{remark}
What we have implicitly used in the proof of Theorem~\ref{thm:main} above  (Equation \eqref{eq:collatzwielandt}) is the following well-known variational characterization of the spectral radius of a positive map $\cT$, known as the \emph{Collatz-Wielandt} formula:
\[
\begin{aligned}
\r(\cT) &= \sup_{w > 0} \sup \{ \lambda \geq 0 : \cT(w) \geq \lambda w \}\\
&= \inf_{w > 0} \inf\{\lambda \geq 0 : \cT(w) \leq \lambda w\}.
\end{aligned}
\]
It holds for any irreducible positive operator $\cT$ and we have denoted $\r(\cT)$ its spectral radius, see e.g., \cite{wolf}. In our setting, we do not need to establish irreducibility as we are simply looking for a bound, and because the ``candidate'' eigenvector $\rho_L$ is positive definite. However to get that $\rho_L$ is close to the eigenvector of $\cL_L^*$ associated to its spectral radius, we do establish irreducibility (and even the stronger property of primitivity) of $\cL_L^*$ for sufficiently large $L$ in the next section, see Lemma~\ref{lem:lbLL}.
\end{remark}

\section{Convergence speed for the eigenvector of \texorpdfstring{$\mathcal{L}_{L}^*$}{} (proof of Theorem~\ref{thm:mainState})}
\label{sec:proof-eigenvector}

In this section, we prove that the eigenvector of $\cL_L^*$ associated to its spectral radius is superexponentially close in trace distance to $\rho_L$, the marginal of the Gibbs state on $[1,L-1]$. The key to obtain this result is to prove that the linear maps $\cL_L^*$ are \emph{primitive}, i.e., that for a sufficiently large integer $k$, $(\cL_L^*)^k = \cL_L^* \circ \dots \circ \cL_L^*$ maps any positive semidefinite operator, to a positive \emph{definite} one. It is well-known in the Perron-Frobenius theory for positive operators \cite{wolf}, that primitivity implies that the spectral radius is a nondegenerate eigenvalue. The next theorem gives a quantitative version of this fact, and shows that a strengthening of the primitivity assumption, implies that approximate eigenvectors (in the sense of \eqref{eq:rhoLapxevec}) are close to the eigenvector associated to the spectral radius. Whereas such a result seems quite natural, we were not able to find it in the literature, so we include here a statement for general positive operators.

\begin{theorem}
\label{thm:apxevector}
Let $\H \cong \CC^n$ be a finite-dimensional Hilbert space, and let $\cT:\cB(\H) \to \cB(\H)$ be a Hermitian-preserving positive linear map, i.e., $\cT(\cB_{+}(\H)) \subset \cB_{+}(\H)$, where $\cB_+(\H)$ is the cone of Hermitian positive semidefinite operators on $\H$.

Assume there exists an integer $k$ such that the following is true: for any $\psi \in \H$, $\cT^k(\psi \psi^*) = \cT \circ \dots \circ \cT(\psi\psi^*) > 0$, and more precisely
\begin{equation}
\label{eq:primquant}
\frac{\lambda_{\max}(\cT^k(\psi \psi^*))}{\lambda_{\min}(\cT^k(\psi \psi^*))} \leq C
\qquad \forall \psi \in \H.
\end{equation}

Then the eigenvector $v \in \cB_+(\H)$ of $\cT$ with $\tr(v)=1$ associated to its spectral radius is positive definite and unique. Furthermore, if $x > 0$ satisfies
\begin{equation}
\label{eq:a1}
(1+\eps)^{-1} r x \leq \cT x \leq (1+\eps) r x
\end{equation}
for some constant $r > 0$ and $\tr x=1$, then $\left\|x - v\right\|_1 \leq 2kC\eps$.
\end{theorem}

In subsection~\ref{sec:proofs1}, we prove Theorem~\ref{thm:apxevector} using as a main tool the \emph{Hilbert projective metric}. In subsection~\ref{sec:proofs2} we show that the maps $\cL^*_L$ satisfy the assumption \eqref{eq:primquant} and use this to complete the proof of Theorem~\ref{thm:mainState}.

\subsection{Proof of Theorem~\ref{thm:apxevector}}
\label{sec:proofs1}

We start by recalling the definition of the Hilbert projective metric on $\cB_{+}(\H)\setminus\{0\}$. We adopt notations from \cite{reeb2011hilbert}.

Given $x,y \in \cB_{+}(\H)$, $x,y\neq0$ we let 
\[
\begin{aligned}
\sup(x/y) &= \inf \{ \lambda \ge 0 : x\leq \lambda y \} \\
\inf(x/y) &= \sup \{ \lambda \ge 0: \lambda y \leq x\}  
\end{aligned}
\]
Note that $\inf(x/y) = 1/\sup(y/x)$. The Hilbert metric is defined by
\[
d_H(x,y) = \log\left(\frac{\sup(x/y)}{\inf(x/y)}\right).
\]
It is easy to check that the following properties are satisfied:
\begin{itemize}
\item $d_H(x,y) = d_H(y,x)$
\item $d_H(\lambda x, \mu y) = d_H(x,y)$ for any $\lambda, \mu > 0$
\item Triangle inequality: $d_H(x,z) \leq d_H(x,y) + d_H(y,z)$
\item Quasi-convexity (see e.g., \cite[Lemma 6.2]{kohlberg1982contraction}): $d_H(x+y,z) \leq \max(d_H(x,z),d_H(y,z))$
\end{itemize}

If $\cT$ is a positive map, it is immediate to check that $\sup(\cT x/\cT y) \leq \sup(x/y)$ and $\inf(\cT x/\cT y) \geq \inf(x/y)$, which implies $d_H(\cT x,\cT y) \leq d_H(x,y)$ for any $x,y > 0$. For primitive maps, one can show that $\cT$ is in fact a contraction for the Hilbert metric.

\begin{theorem}[{Birkhoff, see e.g., \cite[Theorem 4]{reeb2011hilbert}}]
Let $\cT:\cB(\H)\to \cB(\H)$ be a positive map and let
\begin{equation}
\label{eq:DeltaT}
\Delta(\cT) = \sup_{\substack{x,y \geq 0\\ x,y \neq 0}} d_H(\cT x,\cT y).
\end{equation}
Then for any $x,y \geq 0$, 
\begin{equation}
\label{eq:contract}
d_H(\cT x,\cT y) \leq \tanh(\Delta(\cT)/4) d_H(x,y).
\end{equation}
\end{theorem}

It is easy to check that if $\cT$ is primitive and satisfies assumption \eqref{eq:primquant}, then $\Delta(\cT) < \infty$. Indeed, we can easily prove
\begin{proposition}
\label{prop:boundDeltaT}
If $\cT:\cB(\H)\to \cB(\H)$ is a positive map satisfying \eqref{eq:primquant}, then $\Delta(\cT) \leq 2\log(C)$.
\end{proposition}
\begin{proof}
Since $d_H$ is quasi-convex in each of its arguments, one can restrict $x$ and $y$ in the definition of $\Delta(\cT)$ to be rank-one. Furthermore, we have by the triangle inequality
\[
d_H(\cT x, \cT y) \leq d_H(\cT x,\1) + d_H(\1,\cT y).
\]
Note that $d_H(a,\1) = \log(\lambda_{\max}(a) / \lambda_{\min}(a))$, and so equation \eqref{eq:primquant} says that $d_H(\cT x, \1) \leq \log(C)$ for any rank-one positive operator $x$. This proves the claim.
\end{proof}

Since the Hilbert metric is projective, we know that $v$ is a positive eigenvector of $\cT$ if, and only if, $d_H(\cT v,v) = 0$. If $\Delta(\cT) < \infty$, then it is easy to see that $\cT$ can only have one positive eigenvector, for if $v,w$ were two positive eigenvectors of $\cT$, then $d_H(\cT v,\cT w) = d_H(v,w)$ which would contradict \eqref{eq:contract}. One can make this statement more quantitative: if $d_H(\cT x,x) \leq \delta$, then $x$ is necessarily close to the unique positive eigenvector of $T$ in the Hilbert metric.

\begin{proposition}
\label{prop:q1}
Let $\cT$ be a positive map with $\Delta = \Delta(\cT) < \infty$, and let $v$ be its (unique) positive eigenvector. Assume that $x > 0$ satisfies $d_H(\cT x,x) \leq \delta$. Then $d_H(x,v) \leq \delta / (1-\tanh(\Delta/4))$.
\end{proposition}
\begin{proof}
This is a simple triangle inequality. We have
\begin{align*}
d_H(x,v) &\leq d_H(x,\cT x) + d_H(\cT x,v) \\
&= d_H(x,\cT x) + d_H(\cT x,\cT v) \leq d_H(x,\cT x) + \tanh(\Delta/4) d_H(x,v)
\end{align*}
and so $d_H(x,v) \leq \delta / (1-\tanh(\Delta/4))$ as desired.
\end{proof}

We are now ready to prove Theorem~\ref{thm:apxevector}.

\begin{proof}[Proof of Theorem~\ref{thm:apxevector}]
By iterating \eqref{eq:a1} $k$ times we get
\[
r^k (1+\eps)^{-k} x \leq \cT^k x \leq r^k (1+\eps)^k x.
\]
This tells us that $d_H(\cT^k x, x) \leq 2k \log(1+\eps)$. Applying Prop.~\ref{prop:q1} with $\cT^k$ we get that
\begin{equation}
\label{eq:hxv1}
d_H(x,v) \leq \frac{2k \log(1+\eps)}{1-\tanh(\Delta(\cT^k)/4)} \leq 2kC\log(1+\eps).
\end{equation}
where in the last inequality we used Prop.~\ref{prop:boundDeltaT}, which asserts that $\Delta(\cT^k) \leq 2\log C$, which implies $1-\tanh(\Delta(\cT^k)/4) \geq 2/(C+1) \geq 1/C$.

Inequality \eqref{eq:hxv1} shows that $x$ and $v$ are close to each other in the Hilbert projective metric. It remains to prove a similar bound with the trace distance. We assume henceforth that $\tr x = \tr v = 1$.

Let $\eta = 2kC\log(1+\eps)$. Equation \eqref{eq:hxv1} tells us that
\[
c e^{-\eta/2} v \leq x \leq c e^{\eta/2} v
\]
for some constant $c > 0$. By taking the trace, and using the fact that $\tr x = \tr v = 1$, we see that $e^{-\eta/2} \leq c \leq e^{\eta/2}$ and so the inequality above implies
\[
e^{-\eta} v \leq x \leq e^{\eta} v.
\]
This implies that $(e^{-\eta}-1) v \leq x-v \leq (e^{\eta}-1) v$. Since $\eta = 2kC \log(1+\epsilon)$ and $2kC \geq 1$, it is easy to check that $e^{\eta}-1 \leq 2kC \eps$ and $e^{-\eta}-1 \geq -2kC \eps$. This gives us $-2kC \eps v \leq x - v \leq 2kC \eps v$ and this implies that $\|x-v\|_1 \leq 2kC \eps$.\footnote{Indeed, if $-B \leq A \leq B$ then necessarily $\|A\|_1 \leq \tr B$. This follows by noting that if $A = \sum_{\lambda} \lambda P_{\lambda}$ is a spectral decomposition of $A$, then $-\tr(BP_{\lambda}) \leq \tr(AP_{\lambda}) \leq \tr(BP_{\lambda})$ and so $\|A\|_1 = \sum_{\lambda} |\tr(AP_{\lambda})| \leq \sum_{\lambda} \tr(B P_{\lambda}) = \tr(B)$.}
\end{proof}

\subsection{Proof of Theorem~\ref{thm:mainState}}
\label{sec:proofs2}

To apply Theorem~\ref{thm:apxevector} to our map $\cL^*_L$, we need to show that condition \eqref{eq:primquant} holds.
In particular, we need a lower bound on the lowest eigenvalue of the iterated channel applied to some rank-1 projector.
We allow for a number of channel iterations that scales linearly with $L$ and a lower bound that decays exponentially in $L$.

\begin{lemma}\label{lem:lbLL}
There exist $\Pone$-constants $\cG,\cG',\cG'',\cG'''$, such that for $L\ge\exp(\cG'')$ and for all $\psi\in\cH_{[1,L-1]}$ with $\|\psi\|=1$ we have
\[
\cL^{*\cG L}_L(\psi\psi^*)\ge \frac{1}{\cG'} e^{-\cG''' L}\1
\]
\end{lemma}

\begin{remark}\label{rem:conjectureprimitivity}
We believe that Lemma~\ref{lem:lbLL} can be improved. In particular, it is natural to think that $\cL_L^*$ is primitive for all $L \geq 2$. In fact, we conjecture that for any $L \geq 2$, and any $\psi \in \cH_{[1,L-1]}$ such that $\|\psi\|=1$, $(\cL_L^*)^{L-1}(\psi \psi^*) \geq e^{-O(L)} \1$.
If this conjecture is true, then this will improve the dependence on $\|h\|$ in the complexity estimate \eqref{eq:complexityevector}.
\end{remark}

Our proof of Lemma~\ref{lem:lbLL} will require some additional results from \cite{araki1969} concerning the limit of the maps $\cL_L^*$ and their adjoint, which we recall now. 

\paragraph{Limit of $\cL_L$} Note that the adjoint of the map $\cL_L^*$ is given by
\begin{equation}
\label{eq:defLL}
\cL_L(Q) =  \tr_1(E_L^{\dagger}(Q\otimes\id)E_L).
\end{equation}
The key additional fact that we will require from \cite{araki1969} is that if $\cL$ is the limit of the maps $(\cL_L)$, then the iterations $\cL^n$ converge, as $n\to \infty$ to a ``rank-one'' operator (property (iii) below). More precisely, we have\footnote{We use the notation $\cA_{\mathbb{N}}$ for the closure, with respect to the operator norm, of the union of all local algebras $\cA_{[1,n]}$, where operators in $\cA_{[1,n]}$ are embedded into $\cA_{[1,m]}$ for $m > n$ by tensoring with the identity.}:
\begin{lemma}[\cite{perezgarcia2020} Theorem 4.4, 4.7 and 4.12]\label{lem:arakiL}
There exists a map $\cL:\cA_{\mathbb N}\to\cA_\mathbb{N}$, a positive operator $g\in\cA_\mathbb{N}$, a $\Pone$-constant $\cG$, and a superexponentially decaying function $\varepsilon(L)$ such that for all $L$ and $X\in\cA_{[1,L-1]}$
\begin{enumerate}[label=(\roman*)]
\item $\|\cL(X)\|,\|\cL_L(X)\|\le\cG\|X\|$
\item $\|\cL(X)-\cL_L(X)\|\le\varepsilon(L)\|X\|$
\item $\|e^{nf}\cL^n(X)-\tr[\rho_L X]g\|\le \cG e^{-n/\cG} 2^{L+1}\|X\|$
\item $\|g\|,\|g^{-1}\|\le\cG$
\end{enumerate}
\end{lemma}

We are now ready to prove Lemma~\ref{lem:lbLL}.

\begin{proof}[Proof of Lemma~\ref{lem:lbLL}]
Throughout the proof $\cG_i$ are $\Pone$-constants.
As a first step we note that we can equivalently prove the statement for the adjoint channel:
\begin{align*}
&\forall \psi\in\cH_{[1,L-1]} \textrm{ with }\|\psi\|=1,\quad \cL^{*\cG L}_L(\psi\psi^*)\ge a\id\\
&\Leftrightarrow \forall \psi,\phi\in\cH_{[1,L-1]}, \; \|\psi\|=\|\phi\|=1, \; \langle\cL^{*\cG L}_L(\psi\psi^*),\phi\phi^*\rangle\ge a\\
&\Leftrightarrow \forall \psi,\phi\in\cH_{[1,L-1]}, \|\psi\|=\|\phi\|=1, \; \langle\psi\psi^*,\cL_L^{\cG L}(\phi\phi^*)\rangle\ge a\\
&\Leftrightarrow\forall \phi\in\cH_{[1,L-1]} \textrm{ with }\|\phi\|=1,\quad \cL_L^{\cG L}(\phi\phi^*)\ge a\id.
\end{align*}
Now using $\cL$ and $g$ from Lemma~\ref{lem:arakiL} and introducing an integer $M$ that we choose later,
\begin{align*}
\cL_L^M(\phi\phi^*)&=\cL^M(\phi\phi^*)+(\cL^M_L(\phi\phi^*)-\cL^M(\phi\phi^*))\\
&=e^{-fM} (\tr[\rho_L\phi\phi^*]g+ (e^{fM}\cL^M(\phi\phi^*)-\tr[\rho_L\phi\phi^*]g)+e^{fM}(\cL^M_L(\phi\phi^*)-\cL^M(\phi\phi^*)))\\
&\ge e^{-\|h\|M} \1(\cG_1^{-1} e^{-2\|h\|L} - \cG_2 e^{-M/\cG_2} 2^{L+1}-\cG_3^M\varepsilon(L)),
\end{align*}
where the choice of $M$ will ensure that the term in brackets is positive.
In the last step we applied Lemma~\ref{lem:cond} and Lemma~\ref{lem:arakiL} (iv) to lower bound the first term.
The second term is due to Lemma~\ref{lem:arakiL} (iii).
The last term is bounded by a superexponentially decaying function $\varepsilon(L)$, which follows from the following telescope sum:
\begin{align*}
\left\|\cL_L^M(\phi\phi^*)-\cL^M(\phi\phi^*)\right\|&=\left\|\sum_{k=0}^{M-1}\cL^k(\cL_L(\cL_L^{M-k-1}(\phi\phi^*))-\cL(\cL_L^{M-k-1}(\phi\phi^*)))\right\|\\
&\le M \cG_4^M \varepsilon'(L).
\end{align*}
We used Lemma~\ref{lem:arakiL}, which proves the existence of a $\Pone$-constant $\cG_4$ and a superexponentially decaying function $\varepsilon'$ such that $\|\cL_L(Q)\|,\|\cL(Q)\|\le \cG_4\|Q\|$ and $\|\cL_L(Q)-\cL(Q)\|\le\varepsilon'(L)$. 
The factor $M$ can also be upper bounded by $\cG_4^M$, which we can combine into a new $\Pone$-constant $\cG_3=\cG_4^2e^f$.

We choose $M=\left\lceil\cG_2(2\|h\|+\log(4))\right\rceil L:=\cG L$ and express the function $\varepsilon$ explicitly to obtain
\begin{align*}
\cL_L^M(\phi\phi^*)\ge &\  e^{-\|h\|M}\cG_1^{-1}\1\\&\times\left(e^{-2\|h\|L}-2\cG_1\cG_2 e^{-2\|h\|L}\left(\frac12\right)^L-\cG_1 e^{(\cG\log(\cG_5)+C_3(\|h\|+\log(d))-\frac12\log(L/2))L}\right).
\end{align*}
where $C_3$ is a numerical constant.
We now impose that
\[
L\ge\log_2\left(8\cG_1\cG_2\right)
\]
to bound the second term in brackets by $\frac14 e^{-2\|h\|L}$ and similarly 
\begin{align*}
L&\ge2\exp(2(\cG\log(\cG_5)+C_3(\|h\|+\log(d))+2\|h\|+\log(2)))\\
L&\ge\log_2(4\cG_1)
\end{align*}
to ensure the same bound for the last term.
We combine the above bounds into a new $\Pone$-constant $\cG''$ such that $L\ge\exp(\cG'')$ implies all the above inequalities.
Under this condition we conclude
\[
\cL^{\cG M}_L(\phi\phi^*)\ge e^{-\|h\|M}\frac1{2\cG_1}e^{-2\|h\|L}\ge\frac1{\cG'}e^{-\cG''' L}
\]
with $\cG'=2\cG_1$ and $\cG'''=3\|h\|\cG$.

\end{proof}

\begin{proof}[Proof of Theorem~\ref{thm:mainState}]
We consider again $\Pone$-constants $\cG_i$.
We know from Lemma~\ref{lem:arakiL} that $\cL^*_L(Q) \leq \cG_1 \|Q\| \1$, and so this implies that for the iterated map $(\cL^*_L)^{\cG_2L}(Q) \leq \cG_1^{\cG_2L} \|Q\| \1$. 
Together with Lemma~\ref{lem:lbLL} this implies that for any pure state $Q=\psi \psi^*$,
\[
\frac{\lambda_{\max}((\cL^*_L)^{\cG_2L}(Q))}{\lambda_{\min}((\cL^*_L)^{\cG_2L}(Q))} \leq \cG_3 e^{\cG_4 L}.
\]
Furthermore, from Corollary~\ref{cor:LFPSandwich} we know that there is a superexponentially decaying function $\eps(L)$ such that
\[
(1+\eps(L))^{-1} e^{-f} \rho_L \leq \cL_L^*(\rho_L) \leq (1+\eps(L)) e^{-f} \rho_L.
\]
Applying Theorem~\ref{thm:apxevector}, we get that for $L \geq \exp(\cG'')$, $\|v_L - \rho_L\|_1 \leq 2 \cG_2 L \cG_3 e^{\cG_4 L} \eps(L) \leq e^{\cG L} \eps(L)$ for some $\Pone$-constant $\cG$ and the superexponentially decaying function $\eps(L)$.
\end{proof}

\section{Proof of Corollary~\ref{cor:main}}
\label{sec:proof-runtime}

\begin{proof}[Proof of Corollary~\ref{cor:main}]
We start with the algorithm for the free energy. The argument is completely analogous for the Gibbs state. Theorem~\ref{thm:main} tells us that the error incurred by estimating the free energy via the spectral radius of $\cL_L^*$ is at most
\begin{equation}
\label{eq:errorLL}
\cG \frac{\exp(C' L)}{(L/2)!} \; ,
\end{equation}
where $\cG$ is a $\Pone$-constant and $C'$ is of the form $C_3\|h\|$ from the theorem. We need to find a value of $L$ for which the error above is guaranteed to be less than the desired error $\eps$.

As the expression in \eqref{eq:errorLL} cannot be inverted analytically to get $L$ in terms of $\eps$, we use Stirling's approximation to upper bound \eqref{eq:errorLL}. Thus a sufficient condition on our $L$ is that it satisfies:
\[
\varepsilon \ge \cG \exp\left(\left(C-\frac12\log\frac L2\right) L\right).
\]
Note that we replaced $C'$ by another constant $C=C_3(\|h\|+1)$ assuming without loss of generality $C_3>1$ to account for the factor $e$ in the Stirling approximation we used: $n!>(n/e)^n$ (see \cite[Lemma B.3]{friedli2017}).

Representing the concave function in the exponent on the right-hand side as a minimum over linear functions, one can show that this is implied if the condition
\begin{align*}
\log\left(\frac{\varepsilon}\cG\right)\ge \exp\left(2\left(C+\gamma-\frac12\right)\right)-\gamma L
\end{align*}
is fulfilled for some $\gamma$. Using $\gamma=\log(\log(1/\varepsilon))/2$, one can check that the value
\[
L=\left\lceil \frac{\log(1/\varepsilon)(2+2\exp(2C-1))+2\log(\cG)}{\log(\log(1/\varepsilon))}\right\rceil
\]
will work.

The algorithm is now given by the computation of the spectral radius of a square matrix of size $d^{2(L-1)}\times d^{2(L-1)}$, which can be done in time polynomial in its size.
Making the constants explicit again, using the assumption $\varepsilon<1/e$ to include the $\log(\cG)$ in the prefactor of $\log(1/\varepsilon)$, and absorbing all numbers into the new constant $C_1>0$ the resulting runtime then reads
\[
\exp\left(\log(d)\left( \frac{\log(1/\varepsilon)\exp(\exp\left(C_1(\|h\|+1)\right))}{\log(\log(1/\varepsilon))}\right)\right).
\]
For the Gibbs state approximation we simply replace $C'$ by a $\Pone$-constant.
By the exact same calculation we obtain for some positive constants $C_1,C_2>0$
\[
\exp\left(\log(d)\left( \frac{\log(1/\varepsilon)\exp\left(C_1e^{C_2\|h\|}\right)}{\log(\log(1/\varepsilon))}\right)\right).
\]
In addition, we have to choose $L\ge\exp(\cG)$ for some $\Pone$-constant $\cG$ for Theorem~\ref{thm:mainState} to hold and furthermore $L$ has to be at least as big as the size of the marginal that we want to obtain, which is another input to our algorithm.
\end{proof}

\section{Numerical Results}\label{sec:numerics}
We implemented our algorithm and tested it on a dimerized XY model which is exactly solvable \cite{bulaevskii1963,bursill1996}.
For $N+1$ spin-$1/2$ particles, the Hamiltonian is given by
\begin{equation}
\label{eq:dimXYmodel}
H=-\beta\sum_{i=1}^{N/2}\left[\left(S^x_{2i-1}S^x_{2i}+S^y_{2i-1}S^y_{2i}\right)+\gamma\left(S^x_{2i}S^x_{2i+1}+S^y_{2i}S^y_{2i+1}\right)\right],
\end{equation}
where the spin operators are given by the Pauli matrices $S^{x,y}=\sigma^{x,y}/2$.
In the case $\gamma=1$, this model is translation-invariant and is given by the interaction term $h=-\frac{\beta}4(\sigma^x\otimes\sigma^x+\sigma^y\otimes\sigma^y)$.
For $\gamma\ne1$ we do not have translation-invariance in the above sense, but the model still becomes translation-invariant by blocking particles, i.e., it is a translation-invariant Hamiltonian with local dimension $d=4$ and the interaction term
\[
h=-\frac{\beta}4\left(\left(\sigma^x\otimes\sigma^x+\sigma^y\otimes\sigma^y\right)\otimes\1\otimes\1+\gamma\1\otimes\left(\sigma^x\otimes\sigma^x+\sigma^y\otimes\sigma^y\right)\otimes\1\right).
\]
Note also that the system is gapless for $\gamma=1$ and gapped otherwise \cite{bursill1996}.
As discussed in the introduction, for gapped systems there exists an efficient algorithm to approximate the ground state energy, while the problem is $\mathsf{QMA}$-hard in general.

We show results for the error of the log-partition function $\beta f_\beta(h)$ and its dependence on temperature and the parameter $L$ in Figure~\ref{fig:gamma1Plots} for $\gamma=1$. The value of the log-partition function itself as well as the expectation value of the energy is depicted in Figure~\ref{fig:gamma2Plot} for $\gamma=2$.
We compute the energy
\[
e_\beta=\lim_{N\to\infty} \frac{\tr[h_{1,2}e^{-H_{[-N,N]}}]}{\tr[e^{-H_{[-N,N]}}]}
\]
for the two-sided infinite version of the system using the procedure described in Appendix~\ref{sec:recastH}, but also obtain the same results by using a two-sided version of the map $\cL^*_L$.
In Figure~\ref{fig:MIandCMI}, we show the decay of the mutual information \[
I(A:B)=S(\rho_A)+S(\rho_B)-S(\rho_{AB})
\]
between two particles depending on their distance and the same for the conditional mutual information 
\[
I(A:C|B)=S(\rho_{AB})+S(\rho_{BC})-S(\rho_{ABC})-S(\rho_B)
\]
where the conditioning is on the rest of the marginal obtained from the computation (see \cite{kuwahara2020} for an analytic result on this decay at high temperature).
These were computed using the eigenvector of $\cL^*_L$.
Note that they cannot be obtained efficiently from the derivative method (see Theorem~\ref{thm:mainState}) as the derivative method detailed in Appendix~\ref{sec:fToP} is not efficient for obtaining marginals of many particles.

While the theoretical bounds on $L$ for a given error derived in this paper are impractically large, we observe very accurate results for moderate choices of the parameter.
Also the error for the free energy seems to grow less than exponential  with $\beta$ for the chosen example as opposed to our worst-case estimate, which is triple exponential.
In Figure~\ref{fig:ErrorEvsBeta}, we show the $\beta$-dependency of the error in the energy calculation. The dependency seems to be similar to the one for the free energy calculation.
This is in contrast with the worse error dependence in the estimates we proved for the $k$-particle marginals (see Corollary~\ref{cor:main}).
All computations were done on a laptop computer, where computation of a single datapoint for $d=2$, $L=11$ takes about 13\,s.
\begin{figure}[H]
\centering
\subfloat[]{\includegraphics[width=0.5\textwidth]{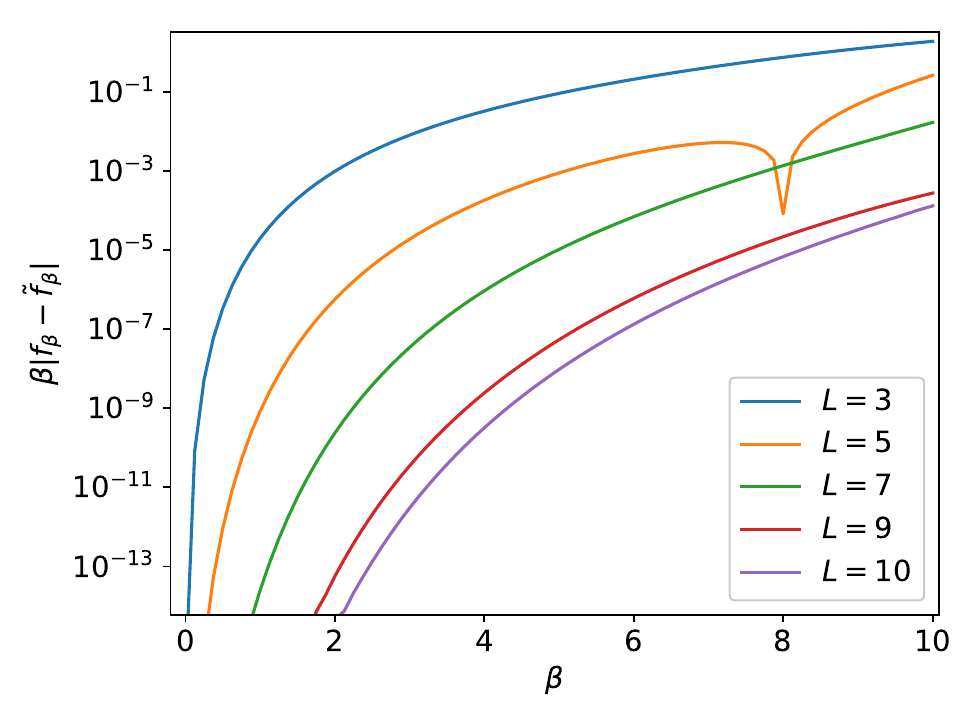}\label{fig:errorPlot}}
\subfloat[]{\includegraphics[width=0.5\textwidth]{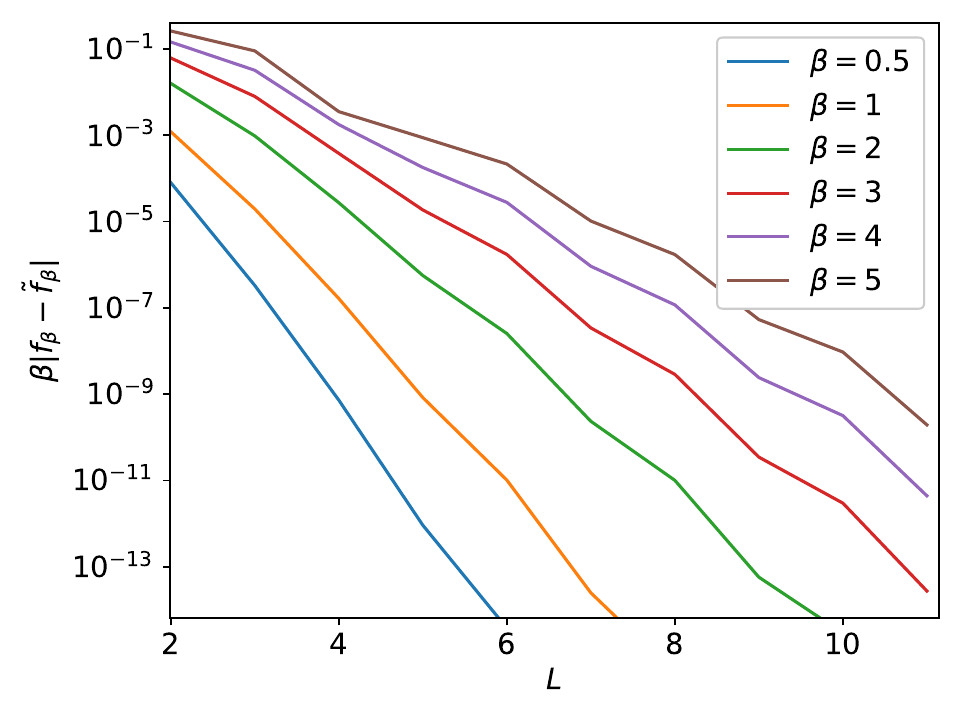}}
\caption{\label{fig:gamma1Plots}(a) Errors on logarithmic scale plotted against the inverse temperature for the model \eqref{eq:dimXYmodel} with $\gamma=1$. The errors decay with $L$ but grow with $\beta$. While our estimates are triple exponential in $\beta$ the shown curves are sublinear suggesting that in practice errors only grow subexponentially with $\beta$. The local minimum for $L=5$ is due to a change of sign in the error. (b) The decay of the error with $L$ for $\gamma=1$.}
\end{figure}
\begin{figure}[H]
\centering
\subfloat[]{\includegraphics[width=0.5\textwidth]{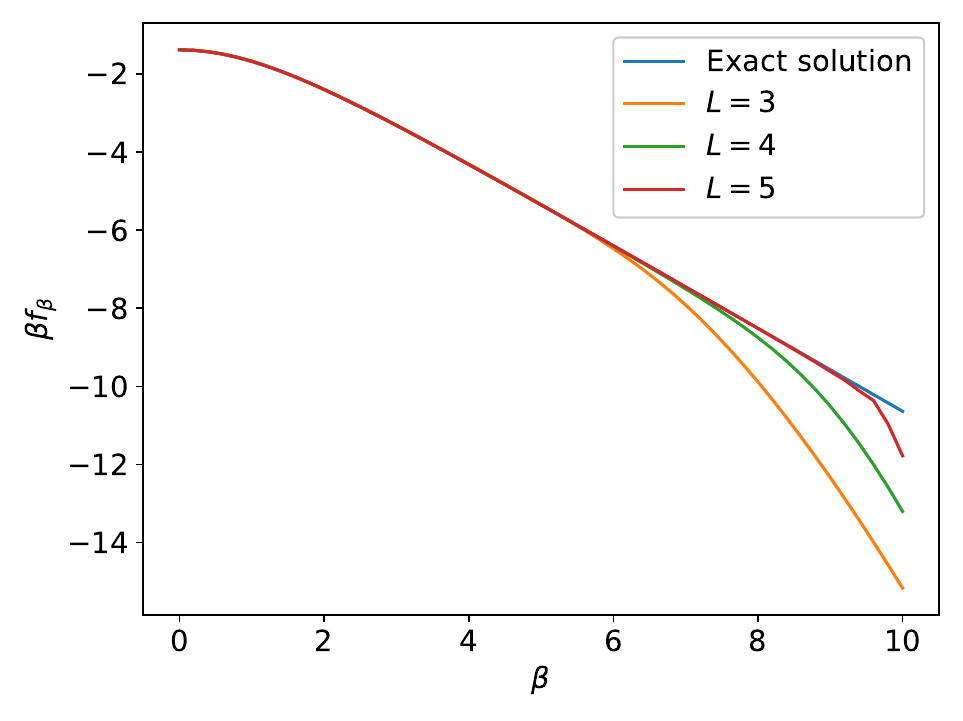}\label{fig:fPlot}}
\subfloat[]{\includegraphics[width=0.5\textwidth]{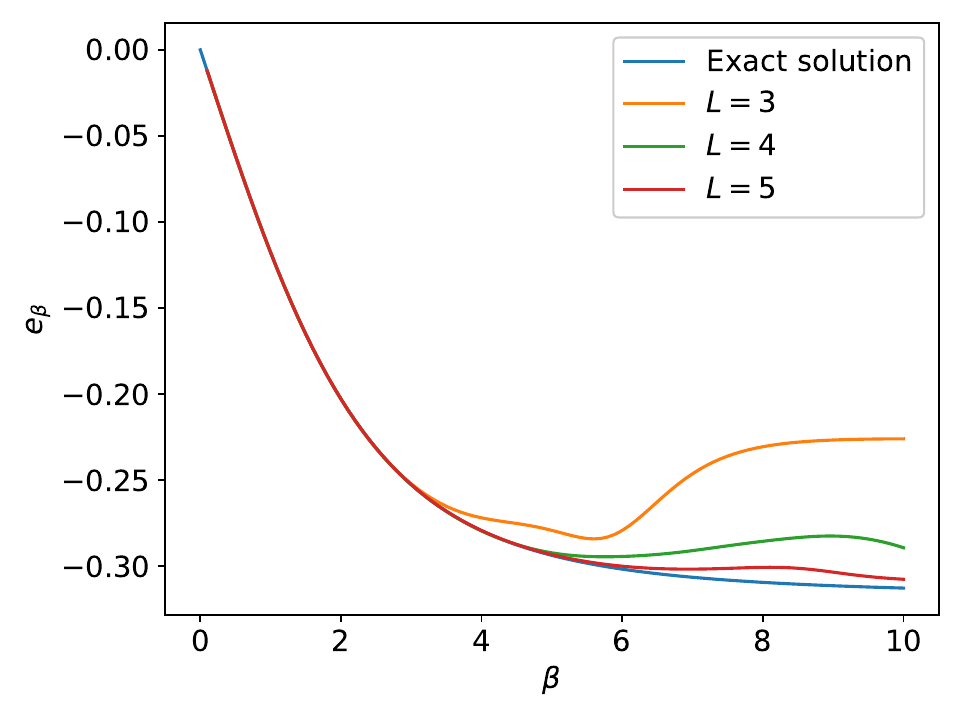}}
\caption{\label{fig:gamma2Plot}
(a) Log-partition function plotted against the inverse temperature for the model \eqref{eq:dimXYmodel} with $\gamma=2$. The curves for all $L$ are not distinguishable from the exact solution for high temperatures but deviate for lower temperatures, where higher $L$ still give better approximations. Due to the choice of $\gamma$ we use a blocking procedure resulting in a local dimension $d=4$, which restricts us to smaller values for the parameter $L$ compared to the case $\gamma=1$ of Figure \ref{fig:gamma1Plots}.
(b) Energy plotted against temperature for $\gamma=1$.
Again, we obtain accurate results at high temperatures.
}
\end{figure}
\begin{figure}[H]
\centering
\includegraphics[width=0.5\textwidth]{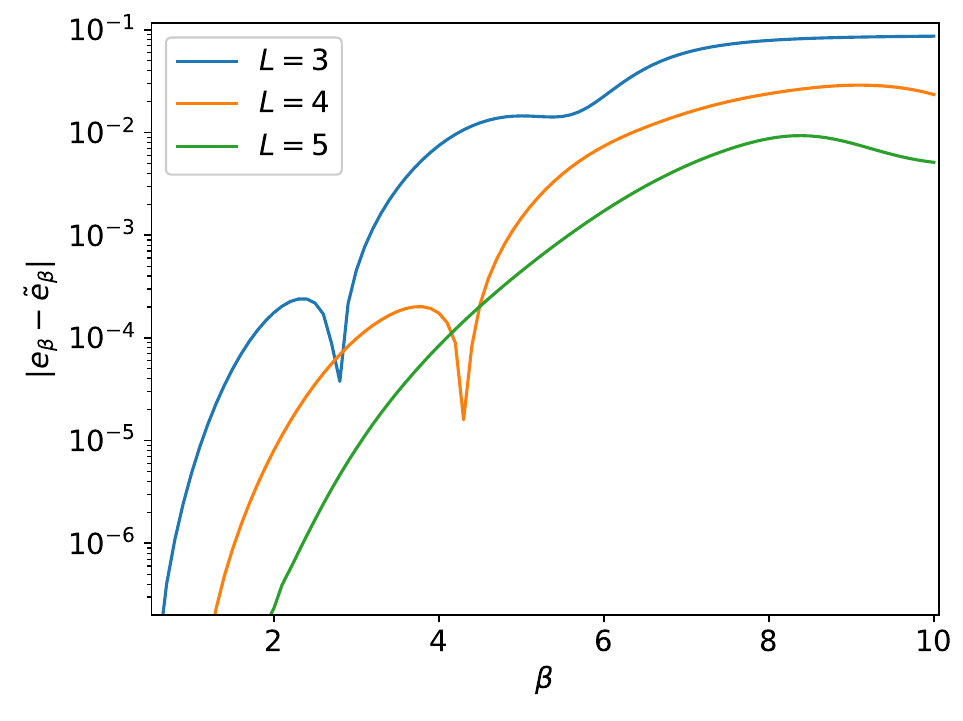}
\caption{\label{fig:ErrorEvsBeta} The error of the energy at $\gamma=1$ when comparing to the exact solution. The dips again arise from a change of sign of the error. The growth with $\beta$ does not seem to be worse than for the free energy.}
\end{figure}
\begin{figure}[H]
\centering
\subfloat[]{\includegraphics[width=0.5\textwidth]{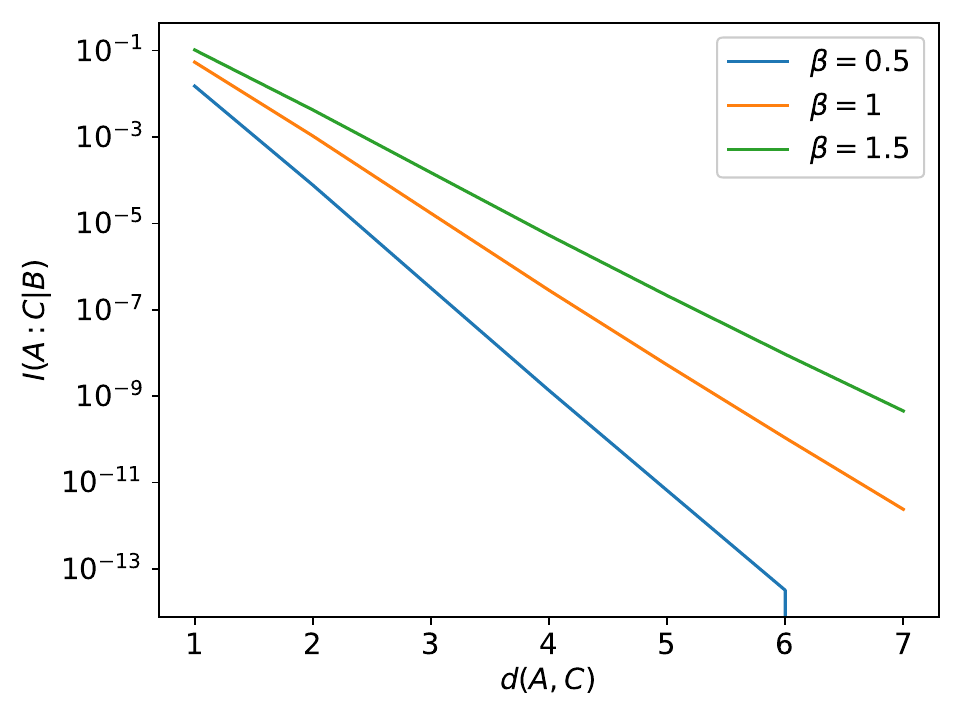}}
\subfloat[]{\includegraphics[width=0.5\textwidth]{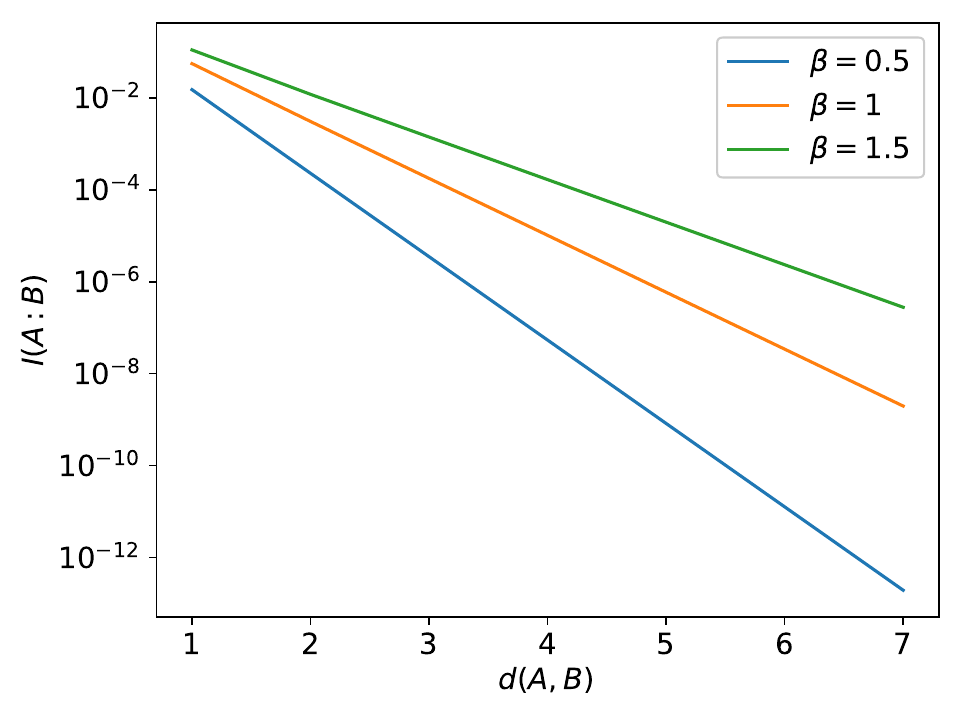}}
\caption{\label{fig:MIandCMI} Decay of the conditional mutual information (a) and the mutual information (b) for the two-sided chain and $\gamma=1$ between two particles with distance $d(\cdot,\cdot)$. The conditioning is on the remaining particles in the 8-site marginal. We observe a growth of the decay rate with temperature.}
\end{figure}

\section*{Acknowledgements}
We thank {\'A}lvaro Alhambra for useful discussions, Arsalan Motamedi for pointing out an error in Lemma \ref{lem:fRatioLimit}, and the anonymous referees for their careful assessment and helpful feedback.
HF acknowledges funding from UK Research and Innovation (UKRI) under the UK
government’s Horizon Europe funding guarantee EP/X032051/1.
OF acknowledges funding by the European Research Council (ERC Grant AlgoQIP, Agreement No. 851716). SOS acknowledges support from the UK Engineering and Physical Sciences Research Council (EPSRC) under grant number EP/W524141/1.

\bibliographystyle{quantum}
\bibliography{refs}
\appendix
\section{A note on the constants in \texorpdfstring{\cite{perezgarcia2020}}{}}\label{sec:constants}
We restated multiple results from \cite{perezgarcia2020}, which form our main technical tools in this work.
While in the reference, constants are mostly left with an implicit dependence on $\| h \|$, we aim to make this explicit in our bounds.
In this appendix we review some of these constants to show that they are bounded is given in the main text.
Also note that the paper mostly considers exponentially decaying interactions and mostly just exponential instead of super-exponential decays.
We can however obtain the better results from there by restricting to local Hamiltonians again.
We list the following constants for a 2-local interaction $\Phi$.
\[
\Omega_0=\Omega_1=\Omega_2=2\|h\|
\]
\[
0=\Omega_3=\Omega_4=\cdots
\]
\[\|\Phi\|_\lambda=2\|h\|(1+e^\lambda+e^{2\lambda})\]
\[\Omega^*_k(2)=\sum_{n=1}^\infty \left(\sum_{\substack{\alpha\in\mathbb N^n\\|\alpha|=k}}\prod_{j=1}^k \Omega_{\alpha_j}\right)\frac{2^n}{n!}\]
The following constant appears in several results cited in our main text and is a $\Pone$-constant due to the bound in \cite[Section 4.2]{perezgarcia2020}, which holds when $\lambda\ge2\Omega_0=4\|h\|$:
\[
\cG=\exp\left(\sum_{k\ge1}e^{2\Omega_0k}\Omega_k^*(2)\right)\le \exp(e^{2\|\Phi\|_\lambda})
\]

To obtain the superexponential decay of $\cG_l$ we start by upper bounding $\Omega_k^*(2)$.
To that end a generally useful formula is the Lagrange remainder bound for the tail of the exponential series:
\[
\sum_{k\ge l}\frac{x^k}{k!}=\exp(x)-\sum_{k=1}^{l-1}\frac{x^k}{k!}\le\frac{e^xx^k}{k!}.
\]
Note that the terms in the outer summation in the definition of $\Omega_k^*(2)$ are zero whenever $n<\lceil k/2\rceil$ as for those one of the factors has to be zero.
Furthermore, the number of vectors in $\mathbb N^n$ with entries no larger than $2$ is bounded by $3^n$ and we obtain
\[
\Omega^*_k(x)\le\sum_{n=\lceil k/2\rceil}^\infty3^n(2\|h\|)^n\frac{x^n}{n!}\le e^{6x\|h\|}\frac{(6x\|h\|)^{\lceil k/2\rceil}}{\lceil k/2\rceil!}.
\]
Plugging this into the expression for $\mathcal G_l$ yields
\begin{align*}
\mathcal G_l&\le e^{e^{2\|\Phi\|_\lambda}} \sum_{k\ge l} e^{2\Omega_0k}e^{12\|h\|}\frac{(12\|h\|)^{\lceil k/2\rceil}}{\lceil k/2\rceil!}\\
&\le e^{e^{2\|\Phi\|_\lambda}+12\|h\|}2\sum_{k\ge\lceil l/2\rceil}e^{4\Omega_0k}\frac{(12\|h\|)^k}{k!}\\
&\le e^{e^{2\|\Phi\|_\lambda}+12\|h\|}2\frac{(e^{8\|h\|}12\|h\|)^{\lceil l/2\rceil}}{\lceil l/2\rceil!}
\end{align*}
which gives us the superexponential-decay of this quantity.

We finally need a bound on the constants $K_2$ and $\delta_2$ from \cite[Theorem 4.12]{perezgarcia2020} which relies on \cite[Theorem 4.11]{perezgarcia2020}.
In the later theorem we have
\[
C_2=2K_2
\]
if $r=\frac{\log(a)}{\log(2)}$ and thereby $N=3r=\frac{3\log(a)}{\log(2)}$.
$K_2$ is bounded by
\[
K_2\le2\cG^4\left(1+\sum_{l\ge1}\cG_l2^l\right)\cG^4(2\cG^3\cG_r2^r+2\cG^4)+2\cG^8\left(1+\sum_{l\ge1}\cG_l2^l\right),
\]
which is a $\Pone$-constant if $1+\sum_{l\ge1}\cG_l2^l$ is as well.
This is the case because, looking at the $l$-dependent term in $\cG_l$
\[
1+\sum_{l\ge1}\frac{\left(e^{8\|h\|}12\|h\|\right)^{\lceil l/2\rceil}}{\lceil l/2\rceil!}\le 2\exp(e^{8\|h\|}12\|h\|)
\]
which is $\Pone$ and the coefficient in $G_l$ is $\Pone$ as well.

We now consider the constants in the proof of \cite[Theorem 4.12]{perezgarcia2020} (note that there will be another constant $K_2$ distinct from the one above).
We have
\[
K=K_2=4C_2\||L^r|\|_{1,2}(1+\||h|\|_{1,2}).
\]
$C_2$ is $\Pone$ as shown above.
For the superoperator norm
\[
\||L^r|\|_{1,2}=\sup_{\||Q|\|_{1,2}\le1}\||L^r(Q)|\|_{1,2}=\sup_{\||Q|\|_{1,2}\le1}\|L^r(Q)\|+\sum_{l\ge1}\|L(Q)\|_l,
\]
we apply \cite[Corollary 4.9]{perezgarcia2020} to bound
\[
\||L^r|\|\le\sup_{\||Q|\|_{1,2}}\cG^4\|Q\|+\sum_{l\ge1}2\cG^3\cG_l\|Q\|+\sum_{l\ge1}\cG^4\|Q\|_{n+l}
\]
where the first sum is $\Pone$ due to the decay of the $\cG_l$ as before and the second sum is upper bounded by $\cG^4\||Q|\|_{1,2}\le\cG^4$ making the whole expression $\Pone$.

A bound for $\||h|\|_{1,2}$ can be deduced from \cite[Theorem 4.10]{perezgarcia2020}, which implies membership of $h$ in certain "quasilocal" sets.
Namely, equation (42) implies $\|h\|_l\le2\cG^3\cG_l$, which we can sum to obtain
\[
\||h|\|_{1,2}\le\|h\|+\sum_{l\ge1}\|h\|_l2^l
\]
which is now $\Pone$ as the sum converges as before. This closes the argument for $K$.

The constant $\delta$ is given as
\[
\delta=-\log\left(1-\frac{1}{2C_2}\right)/N
\]
with $C_2$ and $N=\frac3{\log(2)}\log(a)=\frac3{\log(2)}\log(C_2)$
We estimate using the concavity of the logarithm and its derivative
\[
-\log\left(1-\frac1{2C_2}\right)=\log(2C_2)-\log(2C_2-1)\ge\frac1{2C_2},
\]
so finally
\[
\delta\ge\frac1{6C_2\log(C_2)/\log(3)}
\]
and the denominator is bounded by a $\Pone$-constant.

\section{Reduction of Local Observables to Free Energy}\label{sec:fToP}

In this section we show how one can compute the expectation value of a 2-local observable $P_{1,2} \in \cA_{[1,2]}$ in the thermodynamic limit, by computing the free energy of an appropriately perturbed Hamiltonian.
More precisely, the quantity we are interested in is
\begin{equation}\label{eq:observableLimit}
\mu=\lim_{N\to\infty}\frac{\tr\left[P_{1,2}\exp(-H_{[-N,N]})\right]}{\tr[\exp(-H_{[-N,N]})]},
\end{equation}
where we assume $\|P\|\le 1$.
Our discussion is to some extent similar to the argument in \cite[Lemma 11]{bravyi2021} and adapts it to the infinite translation-invariant case.
The main idea is that a thermal expectation value can be written as a derivative of the free energy (or the partition function) with respect to a parameter introduced in the Hamiltonian.
The derivative can then be approximated by a finite difference with an error that can be bounded by the second derivative.

The case of infinite systems presents a particular challenge compared to the finite case, specifically because the free energy can be nonanalytic in the thermodynamic limit, which gives rise to \emph{phase transitions}. One important consequence of Araki's result \cite[Lemma 9.3]{araki1969} however, is that this does not happen in one-dimensional finite-range systems.
We refine this result and provide a quantitative version, i.e., we establish a bound on the second derivative of the free energy.

This then allows us to quantitatively bound the error caused by a finite difference approximation of the derivative.
Combining this approach with our Algorithm \ref{algo:main} to compute the free energy, we prove the following result.
\begin{lemma}\label{lem:fRuntimeViaDerivative}
For some numerical constants $C_1$, $C_2 > 0$, there exists an algorithm that takes as input the local dimension of a quantum system $d$, its 2-local translation invariant Hamiltonian $h$, a 2-local normalized observable $P$, an additive error $\eps$, runs in time at most
\[
\exp\left(\log(d)\exp(\exp(C_1e^{C_2\|h\|}))\frac{\log(1/\eps)}{\log(\log(1/\eps))}\right),
\]
and outputs an approximation $\tilde\mu$ such that $|\mu-\tilde\mu|\le\eps$, where $\mu$ is defined in \eqref{eq:observableLimit}.
\end{lemma}

Note that we restrict to 2-local observables as our algorithm for the free energy (Algorithm \ref{algo:main}) only applies to 2-local Hamiltonians supported on two neighbouring sites.
The more general case of $k$-local observables (again supported on $k$ contiguous sites) follows by first blocking the system. We note however that the resulting algorithm will have a bad dependence on $k$ in its runtime, as the local dimension of the blocked Hamiltonian will be $d^{k/2}$ and its interaction strength can be of the order of $k\|h\|/2$.

\subsection{Preliminaries}

For the proof of Lemma \ref{lem:fRuntimeViaDerivative}, we will need some additional results from \cite{perezgarcia2020} and \cite{bluhm2022exponential} concerning the quasi-locality of the time evolution operator, and decay of correlations. We introduce the necessary notations here.

Given a 2-local Hamiltonian, we define the time-evolution operator
\[
\Gamma_{[a,b]}^s(A)=e^{isH_{[a,b]}}Ae^{-isH_{[a,b]}}.
\]
Also, we define an $l$-neighbourhood of a particle
$j\in [1,N]$ as $\Lambda_j^l=[\max\{1,j-l\},\min\{N,j+1+l\}]$.

We also need the following notation for the thermal state of the finite system on $[a,b]$:
\[
\rhot_{[a,b]} = \frac{e^{-H_{[a,b]}}}{\tr\left[e^{-H_{[a,b]}}\right]}.
\]
The results needed for our proof are summarized in the following Lemma.
\begin{lemma}[{{\cite[Theorems 4.16 and 2.4]{perezgarcia2020} and  \cite[Theorem 6.2]{bluhm2022exponential}}}]\label{lem:PGPHBluhm}
Consider a 2-local translation invariant Hamiltonian of the form $H_{[a,b]} = \sum_{i=a}^{b-1} h_{i,i+1}$. Then the following holds:
\begin{enumerate}[label=(\roman*)]
    \item There exist constants $C, \delta > 0$, such that for any normalized operator $P_{[a,b]}\in\cA_{[a,b]}$, and any integer $n \geq 1$, we have \[
    \left|\tr\left[P_{[a,b]} \rhot_{[a-n,b-n]}\right]-\lim_{N\to\infty}\tr\left[P_{[a,b]} \rhot_{[-N,N]}\right]\right|\le C e^{-\delta n}.\]
    \item There is a superexponentially decaying function $\eta(l)$, such that for any normalized observable $P_{j,j+1}\in \cA_{[j,j+1]}$
 \[
 \left\|\Gamma^{is}_{[1,N]}(P_{j,j+1})-\Gamma_{\Lambda_j^l}^{is}(P_{j,j+1})\right\|\le\eta(l).
 \]
 \item There are $\Pone$-constants $\cG$, $K$, $1/\alpha$ such that for any observables $O_A$, $O_B$ with supports separated by $l$ sites we have
 \[
 \left|\tr[O_AO_B\rhot_{[1,N]}]-\tr[O_A\rhot_{[1,N]}]\tr[O_B\rhot_{[1,N]}]\right|\le\|O_A\|\|O_B\|\left(Ke^{-\alpha l}+\cG\frac{\cG^l}{(l/2)!}\right).
 \]
\end{enumerate}
\end{lemma}

\subsection{Proof of Lemma \ref{lem:fRuntimeViaDerivative}}

The algorithm in the proof of Lemma \ref{lem:fRuntimeViaDerivative} is based on computing the derivative of the free energy function with respect to a parameter in a perturbed Hamiltonian. The perturbed translation-invariant Hamiltonian we consider is defined by
\[
H_{[1,N]}(\eps)=\sum_{i=1}^{N-1} (h_{i,i+1}+\eps P_{i,i+1}).
\]
We denote the corresponding partition function $Z_{N}(\eps) = \tr e^{-H_{[1,N]}(\eps)}$ and free energy per site $f_N(\eps) = -\frac{1}{N} \log Z_N(\eps)$, as well as the limit
\begin{equation}
\label{eq:feps}
f(\eps) = \lim_{N\to \infty} f_N(\eps).    
\end{equation}
We also make use of the perturbed thermal state associated to the finite system on the interval $[a,b]$ as
\[
\rhot_{\eps,[a,b]}=\frac{\exp(-H_{[a,b]}(\eps))}{\tr\left[\exp(-H_{[a,b]}(\eps))\right]}.
\]
For convenience, we let $\rhot_{\eps,N} = \rhot_{\eps,[1,N]}$.

The first lemma shows that the expectation value $\mu$ defined in \eqref{eq:observableLimit} is equal to the derivative of the free energy for the perturbed Hamiltonian $f(\eps)$ defined in \eqref{eq:feps}.
\begin{lemma}\label{lem:convergenceP}
For any 2-local observable $P \in \cA_{[1,2]}$ we have
\[
\dv{\eps}f(0) = \mu := \lim_{N\to\infty}\frac{\tr\left[P_{1,2}\exp(-H_{[-N,N]})\right]}{\tr[\exp(-H_{[-N,N]})]},
\]
where $f(\eps)$ is the free energy per site of the perturbed Hamiltonian, as defined in \eqref{eq:feps}.
\end{lemma}
\begin{proof}
We assume without loss of generality that $\|P\|=1$.
It is easy to verify that for any $N$,
\[
\dv{\eps}f_N(\eps)=\tr\left[\frac1{N-1}\sum_{i=1}^{N-1}P_{i,i+1}\rhot_{\eps,N}\right].
\]
Define $\mu(\eps) = \lim_{N\to\infty} \tr[P_{1,2} \rhot_{\eps,[-N,N]}]$. Lemma~\ref{lem:PGPHBluhm} (i) tells us that for any $i \in [1,N]$,
\[
\left|\tr[P_{i,i+1} \rhot_{\eps,N}] - \mu(\eps)\right| \leq C e^{-\delta \min(i,N-i)}
\]
for some constants $C, \delta > 0$. It thus follows that
\[
\begin{aligned}
\left|\frac{1}{N-1} \sum_{i=1}^{N-1} \tr[P_{i,i+1} \rhot_{\eps,N}] - \mu(\eps)\right|& \leq \frac{C}{N-1} \sum_{i=1}^{N-1} e^{-\delta \min(i,N-i)}\\
&\leq \frac{2C}{N-1} \sum_{i=1}^{N/2} e^{-\delta i} \leq \frac{2C}{N-1}(M + (N/2-M) e^{-\delta M} )
\end{aligned}
\]
where in the last step $M$ is any integer in $[1,N/2]$. By taking $M=\sqrt{N/2}$ we see that $\dv{\eps}f_N(\eps) \to \mu(\eps)$ as $N\to \infty$.
To prove that $\dv{\eps} f_N \to \dv{\eps} f$, we need the convergence to be uniform on all compact sets.
While the constants $C$ and $\delta$ do depend on the interaction strength $\|h\|+\eps$, an upper bound on these constants grows monotonically with $\|h\|+\eps$.
Therefore, we just choose the constants for interaction strength $\|h\|+\eps'$ to prove uniform convergence for $\eps\in[0,\eps']$.
The convergence of the derivatives then follows from a standard result in analysis \cite[Theorem 7.17]{rudin1965}.
\end{proof}

Before we proceed we need a basic result from calculus. We were not able to find a proof in the literature for this specific setting so we include a proof for completeness.
Note that the statement does not require convergence of the derivatives $g_N'$.
\begin{lemma}\label{lem:calcSecondDeriv}
Let $g_N:(a,b)\to\mathbb R$ be a sequence of continuously differentiable functions converging to a continuously differentiable function $g$.
If the derivatives of the elements are uniformly bounded $g_N'(x)\le C$, then the same bound holds for the limit $g'(x)\le C$.
\end{lemma}
\begin{proof}
Let us assume for contradiction that there exists $x_0\in(a,b)$, such that $g'(x_0)>C$.
Due to the continuity of $g'$ there exist $\delta$, $\eps$ such that $g'(x)>C+\delta$ for all $x\in[x_0, x_0+\eps]$.
We now choose $N$ sufficiently large such that $|g_N(x)-g(x)|\le\eps\delta/3$ for $x\in[x_0,x_0+\eps]$.
Then
\[
\frac{g_N(x_0+\eps)-g_N(x_0)}{\eps}\ge\frac{\eps(C+\delta)-2\eps\delta/3}{2\eps}\ge C+\frac13\delta
\]
By the mean value theorem this implies that there exists $x_1\in[x_0,x_0+\eps]$ such that $g_N'(x_1)\ge C+\frac13\delta$, which contradicts our assumption and closes the proof.
\end{proof}
We can now prove a bound on the second derivative of the free energy using the quasi-locality of complex time-evolved operators and the decay of correlations in the thermal state.
\begin{lemma}\label{lem:secondDerivF}
For each $\eps_0$ there exists a $\Pone$-constant $\cG$ for Hamiltonian strength $\|h\|+\eps_0$ such that
\[
\ddv{2}{\eps}f(\eps)\le e^\cG
\]
for all $\eps\in[0,\eps_0]$.
\end{lemma}
\begin{proof}
Let us consider the parameter derivatives of the free energy density of a finite system

\begin{align*}
f_N(\eps)&=-\frac1N\log\tr\exp(-H_{[1,N]}(\eps))\\
\dv{\eps}f_N(\eps)&=\tr\left[\frac1{N-1}\sum_{i=1}^{N-1}P_{i,i+1}\rhot_{\eps,N}\right]\\
\ddv{2}{\eps}f_N(\eps)&=\frac1{N-1}\left[\frac{\tr\left[\sum_{i=1}^{N-1}P_{i,i+1}\int_0^1 e^{-s H_{[1,N]}(\eps)}\sum_{i=1}^{N-1}P_{i,i+1}e^{-(1-s)H_{[1,N]}(\eps)}ds\right]}{\tr[\exp(-H_{[1,N]}(\eps))]}\right.\\
&\quad-\left(\tr\left[\sum_{i=1}^{N-1}P_{i,i+1}\rhot_{\eps,N}\right]\right)^2\left.\rule{0cm}{0.75cm}\right],
\end{align*}
where the last line can be derived using the Duhamel formula (compare also \cite{bravyi2021}).
We rewrite the second derivative in a slightly more compact form as
\begin{align*}
\ddv{2}{\eps}f_N(\eps)&=\frac1{N-1}\sum_{j,k=1}^{N-1}\tr\left[P_{j,j+1}\int_0^1 \Gamma^{is,\eps}_{[1,N]}(P_{k,k+1})ds\rhot_{\eps,N}\right]-\tr[P_{j,j+1}\rhot_{\eps,N}]\tr[P_{k,k+1}\rhot_{\eps,N}],
\end{align*}
where we used the time evolution operator with respect to the perturbed Hamiltonian
\[
\Gamma_{[a,b]}^{s,\eps}(A)=e^{isH_{[a,b]}(\eps)}Ae^{-isH_{[a,b]}(\eps)}.
\]
We fix $j \in [1,N-1]$ and $s \in [0,1]$ and focus on the sum
\[
\sum_{k=1}^{N-1}\tr\left[P_{j,j+1}\Gamma_{[1,N]}^{is,\eps}(P_{k,k+1})\rhot_{\eps,N}\right] - \tr[P_{j,j+1}\rhot_{\eps,N}]\tr[P_{k,k+1}\rhot_{\eps,N}].
\]
The idea is that the imaginary-time evolved operator $\Gamma_{[1,N]}^{is,\eps}(P_{k,k+1})$ is approximately supported in a region close to $\{k,k+1\}$ and that it is thereby approximately uncorrelated from $P_{j,j+1}$ if $|j-k|$ is large.

We apply Lemma~\ref{lem:PGPHBluhm} (ii) and (iii) and choose $l(k)=\max\{(|j-k|-2)/2,0\}$ to estimate
\begin{align*}
\left|\sum_{k=1}^{N-1}\right.&\left(\tr\left[P_{j,j+1}\Gamma_{[1,N]}^{is,\eps}(P_{k,k+1})\rhot_{\eps,N}\right]-\tr[P_{j,j+1}\rhot_{\eps,N}]\tr[P_{k,k+1}\rhot_{\eps,N}]\right) \Bigg|\\
\le&\ \left|\sum_{k=1}^{N-1}\left(\tr\left[P_{j,j+1}\Gamma_{[1,N]}^{is,\eps}(P_{k,k+1})\rhot_{\eps,N}\right]-\tr\left[P_{j,j+1}\Gamma_{\Lambda_k^{l(k)}}^{is,\eps}(P_{k,k+1})\rhot_{\eps,N}\right]\right)\right|\\
&\ +\left|\sum_{k=1}^{N-1}\left(\tr\left[P_{j,j+1}\Gamma_{\Lambda_k^{l(k)}}^{is,\eps}(P_{k,k+1})\rhot_{\eps,N}\right]\right.\right.-\tr\left[P_{j,j+1}\rhot_{\eps,N}\right]\tr\left[\Gamma_{\Lambda_k^{l(k)}}^{is,\eps}(P_{k,k+1})\rhot_{\eps,N}\right]\Big)\Bigg|\\
&\ +\left|\sum_{k=1}^{N-1}\left(\tr\left[P_{j,j+1}\rhot_{\eps,N}\right]\tr\left[\Gamma_{\Lambda_k^{l(k)}}^{is,\eps}(P_{k,k+1})\rhot_{\eps,N}\right]\right.\right.-\tr[P_{j,j+1}\rhot_{\eps,N}]\tr[P_{k,k+1}\rhot_{\eps,N}]\Big)\Bigg|\\
\le&\ \sum_{k=1}^{N-1}\left(\eta(l(k))+\cG\frac{\cG^{l(k)}}{\lceil l(k)\rceil!} +Ke^{-\alpha(l(k))}+\eta(l(k))\right) \\
\le&\ e^{\cG'}.
\end{align*}
We used the quasi-locality to bound the first and third term after the first inequality (note that $\tr[P_{k,k+1}\rhot_{\eps,N}]=\tr[\Gamma_{[1,N]}^{is,\eps}(P_{k,k+1})\rhot_{\eps,N}]$ as the time evolution operator commutes with the thermal state).
The second term is due to the decay of correlations and the bound on $\left\|\Gamma^{is,\eps}_{\Lambda_j^{l(k)}}(P_{k,k+1})\right\|$ that is also implied by Lemma~\ref{lem:PGPHBluhm} (ii).
Note that the sum over a superexponentially decaying function is $\Pone$ by definition of the exponential series, as is the sum of an exponential when the inverse of the decay rate is $\Pone$.
The second term in the sum is responsible for the exponential bound in a $\Pone$-constant again by definition of the exponential series.

By integrating over $s$, summing over $j$, and dividing by $N-1$ this implies that $\ddv2{\eps} f_N(\eps)$ is bounded uniformly in $N$ by the exponential of a $\Pone$ constant (for the Hamiltonian $H(\eps)$ whose interaction term has norm up to $\|h\|+\eps$). 
Using the continuous second-order differentiability of the limit function $f(\eps)$ from \cite[Lemma 9.3]{araki1969}, the claim follows from Lemma~\ref{lem:calcSecondDeriv} applied to $\dv{\eps}f(\eps)$.
\end{proof}

We now have all the tools to estimate thermal expectation values using an approximation of the free energy.
\begin{proof}[Proof of Lemma~\ref{lem:fRuntimeViaDerivative}]

We are looking for an approximation $\tilde \mu$ such that
\[
\left|\mu-\tilde\mu\right|\le\eps.
\]
We have for the second derivative of the free energy from Lemma~\ref{lem:secondDerivF}
\[
\left|\ddv{2}{\eps} f(\eps)\right|\le e^\cG
\]
for all $\eps\in[0,1]$ and for some constant $\cG$ that is $\Pone$ for the interaction strength $\|h\|+1$.
To achieve the bound we pick $\eps'=\min\left\{1,\frac{\eps}{2\exp(\cG)}\right\}$ and compute
\[
\tilde\mu:=\frac{\tilde f(\eps')-\tilde f(0)}{\eps'},
\]
where we use an error of $|f(\eps')-\tilde f(\eps')|,|f(0)-\tilde f(0)|\le\delta:=\eps\eps'/4$.
Using a standard calculus formula
\[
\left|\dv{\eps}f(0)-\frac{f(\eps')-f(0)}{\eps'}\right|\le \frac{\eps'}2\max_{0\le s\le\eps'}\ddv2s f(s),
\]
we conclude using our bound on the second derivative and the chosen errors
\[
|\mu-\tilde\mu|\le\left|\dv{\eps}f(0)-\frac{f(\eps')-f(0)}{\eps'}\right|+\left|\frac{f(\eps')-\tilde f(\eps')}{\eps'}\right|+\left|\frac{f(0)-\tilde f(0)}{\eps'}\right|\le\eps.
\]

Plugging the desired error into the runtime bounds for the free energy of Corollary~\ref{cor:main}, we conclude that the runtime for the thermal expectation value is bounded by
\[
\exp\left(\log(d)\exp(\exp(C_1e^{C_2\|h\|}))\frac{\log(1/\eps)}{\log(\log(1/\eps))}\right)
\]
for some numerical constants $C_1,C_2$.
\end{proof}
\section{Reducing the Ground Energy Problem to the Free Energy}\label{sec:GSReduction}
We reduce the ground-state problem to the free energy problem to establish $\mathsf{QMA_{EXP}}$-hardness of the following problem.
\begin{definition} The \textit{Free Energy for Infinite Translation-Invariant Hamiltonian} FE-ITIH is defined as follows\\
\textbf{Problem parameter:} Three polynomials $\bar{p}, \bar{q}, \bar{r}$ and $d$ the dimension of a particle \\
\textbf{Problem input:} $N$ specified in binary, $\beta \geq 0$ and the matrix $h$, each specified in binary with at most $\log \bar{r}(N) + 1$ bits \\
\textbf{Promise:} \sloppy The free energy density of the infinite system defined by $h$, i.e., $f_{\beta}(h) = \lim_{N \to \infty} -\frac{1}{\beta N} \log \tr \exp(-\beta H_{[1,N]}) $ and lies outside the interval $[1/\bar{p}(N),1/\bar{p}(N)+1/\bar{q}(N)]$\\
\textbf{Output:} Determine if free energy density is at most $1/\bar{p}(N)$ or at least $1/\bar{p}(N) + 1/\bar{q}(N)$
\end{definition}
\begin{lemma}
There exist parameters $\bar{p},\bar{q}, \bar{r}, \bar{d}$ such that the above problem is $\mathsf{QMA_{EXP}}$-hard.
\end{lemma}
\begin{proof}
We start by giving a quantitative convergence bound for the free energy to the ground-state energy for $\beta\to\infty$.

Using the well-known variational formula for the free energy, we can write the free energy density for a finite system as
\[
f_{\beta, N}(h) =\frac1N\min_{\rho\ge0,\tr[\rho]=1} \left( \tr[H_{[1,N]}\rho]-\frac1\beta S(\rho) \right),
\]
where $S(\rho) = -\tr[\rho \log_2 \rho]$ denotes the von Neumann entropy.
As $0\le S(\rho)\le N\log(d)$ we can bound the free energy as
\[
\min_{\rho\ge0,\tr[\rho]=1}\tr\left[\frac{H_{[1,N]}}{N}\rho\right]-\frac1\beta\log(d)\le f_{\beta, N} \le \min_{\rho\ge0,\tr[\rho]=1}\tr\left[\frac{H_{[1,N]}}{N}\rho\right].
\]
Let $e_{0, N} = \frac{1}{N} \min_{\rho\ge0,\tr[\rho]=1}\tr[H_{[1,N]}\rho]$ so that we have
\[
\left| f_{\beta, N}-e_{0, N}\right|\le \frac{\log(d)}\beta .
\]
The existence of the limit $\lim_{N \to \infty} f_{\beta, N}$ shows the existence of the limit ground-state energy $e_0 = \lim_{N \to \infty} e_{0, N}$. As a result, we have
\[
\left| f_{\beta}-e_{0}\right|\le \frac{\log(d)}\beta.
\]

We now take a $\mathsf{QMA_{EXP}}$-hard problem $(p,q,r,d)$ ITIH as defined in \cite[Theorem 2.7]{gottesman2009}. We set the parameters for the FE-ITIH problem: $\bar{p} = p, \bar{q} = 2 \cdot q, \bar{r} = r + 2 (\log d) q$ and we use the same dimension parameter $d$.
Given an instance $(N, h)$ for the ITIH problem, we will consider the instance $(N, \beta, h)$ where we let $\beta = 2 q(N) \log d$. Note that $\beta \leq \bar{r}(N)$ and thus can be specified using at most $\log \bar{r}(N) + 1$ bits. We now check that the promise is satisfied: we know that $e_0$ is either $\leq \frac{1}{p(N)}$ or at least $\frac{1}{p(N)} + \frac{1}{q(N)}$. Now for the free energy density, as $f_{\beta} \leq e_0$, we have either $f_{\beta} \leq \frac{1}{p(N)}$ or $f_{\beta} \geq \frac{1}{p(N)} + \frac{1}{q(N)} - \frac{\log d}{\beta} = \frac{1}{p(N)} + \frac{1}{2q(N)}$. It then follows that for the instance we constructed, $e_0 \leq \frac{1}{p(N)}$ if and only if $f_{\beta} \leq \frac{1}{p(N)}$ which proves the desired result.
\end{proof}

We can interpret this hardness result by saying that we cannot expect to have an algorithm that has a dependence of the form $\exp(\mathrm{polylog}(\beta, 1/\eps))$, unless $\mathsf{QMA_{EXP}} = \mathsf{EXP}$.

\section{Recasting a two-sided chain to a one-sided chain} \label{sec:recastH}

In this appendix, we explain how the Gibbs state of an infinite two-sided chain can be obtained from the Gibbs state of an infinite one-sided chain for a suitably modified Hamiltonian.

We consider a local dimension $d$ and 2-particle Hamiltonian $h \in \cB(\CC^d \otimes \CC^d)$ as given.
We construct a system where the local Hilbert space corresponds to two copies of the original one, i.e., $\cH' = \CC^{d}\otimes\CC^{d}$ with local dimension $d'=d^2$.
We denote the sites of the system by indices with subscript $i_u$, $i_d$ referring to either of the two copies for particle $i$ respectively.
The Hamiltonian for the one-sided chain $h' \in \cB(\CC^{d'} \otimes \CC^{d'})$ is then chosen as $h'_{1,2}=h_{2_u,1_u}+h_{1_d,2_d}+h_{1_u,1_d}-h_{2_u,2_d}$, see Figure~\ref{fig:recastedH} for an illustration.

In the sum over Hamiltonian terms on all sites, the third and last term occur each once for each particle and thereby cancel.
The only exception is the first particle, where only the positive term appears.
The construction can be thought of as a chain winding at the first site with both infinite ends going in the same direction.
The thermal state on this construction is then equivalent to the two-sided infinite thermal state when using the order of sites as $\rho_{i_u,\cdots,2_u,1_u,1_d,2_d,3_d,\cdots,i_d}$.
\begin{figure}[H]
\centering
\begin{tikzpicture}
\draw[black,thick] (0,0)--(2,0);
\draw[black,thick] (2,0)--(4,0);
\draw[black,thick] (4,0)--(6,0);
\draw[black,thick] (6,0)--(7,0);
\draw[black,thick] (0,0)--(0,-2);
\draw[black,thick] (0,-2)--(2,-2);
\draw[black,thick] (2,-2)--(4,-2);
\draw[black,thick] (4,-2)--(6,-2);
\draw[black,thick] (6,-2)--(7,-2);
\draw[red,thick] (2,0)--(2,-2);
\draw[red,thick] (4,0)--(4,-2);
\draw[red,thick] (6,0)--(6,-2);
\filldraw[black](0,0) circle (5pt) node[anchor=north west]{$1_u$};
\filldraw[black](2,0) circle (5pt) node[anchor=north west]{$2_u$};
\filldraw[black](4,0) circle (5pt) node[anchor=north west]{ $3_u$};
\filldraw[black](6,0) circle (5pt) node[anchor=north west]{ $4_u$};
\filldraw[black](0,-2) circle (5pt) node[anchor=north west]{ $1_d$};
\filldraw[black](2,-2) circle (5pt) node[anchor=north west]{ $2_d$};
\filldraw[black](4,-2) circle (5pt) node[anchor=north west]{ $3_d$};
\filldraw[black](6,-2) circle (5pt) node[anchor=north west]{ $4_d$};
\filldraw[black](7,-1) circle (0pt) node[anchor=west]{\textbf{...}};
\end{tikzpicture}
\caption{\label{fig:recastedH} The nodes in each column correspond to the bipartite new local Hilbert space. Each black line corresponds to a Hamiltonian term, while the red lines represent a zero due to the cancellation from two neighbouring terms.}
\end{figure}
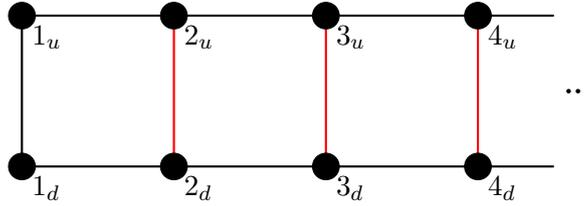

\end{document}